\documentclass[journal]{IEEEtran}

\usepackage{cite}
\usepackage{amsmath,amssymb,amsfonts,amsthm}
\usepackage{algorithmic}
\usepackage{textcomp}

\usepackage{dsfont,eucal,bbm,bm,nicefrac} 
\usepackage{graphicx,float,subcaption,booktabs} 
	\graphicspath{{./figures/}}
\usepackage{algorithm,algorithmic} 
  \usepackage{tikz,xcolor} 

\newtheorem{theorem}{Theorem}
\newtheorem{corollary}{Corollary}[theorem]
\newtheorem{lemma}{Lemma}

\newtheorem{remark}{Remark}

\newtheorem{definition}{Definition}

\newtheorem{example}{Example}

\newcommand{\T}{\ensuremath{\mathrm{T}}} 

\newcommand{\norm}[1]{\ensuremath{\left\| #1\right\|}}
\newcommand{\pbra}[1]{\ensuremath{\left( #1\right)}}
\newcommand{\sbra}[1]{\ensuremath{\left[ #1\right]}}
\newcommand{\cbra}[1]{\ensuremath{\left\{ #1\right\}}}
\newcommand{\abra}[1]{\ensuremath{\left< #1\right>}}

\newcommand{\pder}[2]{\ensuremath{\frac{\partial #1}{\partial #2}}}
\newcommand{\E}[1]{\ensuremath{\mathbb{E}\left[ #1\right]}}

\newcommand{\Eh}[1]{\ensuremath{\mathbb{\hat E}\left[ #1\right]}}

\DeclareMathOperator*{\argmax}{arg\,max}
\DeclareMathOperator*{\argmin}{arg\,min}

\begin{document}

\title{Annealing Optimization for Progressive Learning with Stochastic Approximation}

\author{Christos N. Mavridis, \IEEEmembership{Member, IEEE}, and 
John S. Baras, \IEEEmembership{Life Fellow, IEEE}
\thanks{
The authors are with the 
Department of Electrical and Computer Engineering and 
the Institute for Systems Research, 
University of Maryland, College Park, USA.
{\tt\small emails:\{mavridis, baras\}@umd.edu}.}%
\thanks{Research partially supported by the 
Defense Advanced Research Projects
Agency (DARPA) under Agreement No. HR00111990027, 
by ONR grant N00014-17-1-2622, 
and by a grant from Northrop Grumman Corporation.%
}%
}

\maketitle
 \thispagestyle{empty}
\pagestyle{empty}

\begin{abstract}
In this work, we introduce a learning model designed to
meet the needs of applications in
which computational resources are limited, and robustness and
interpretability are prioritized.
Learning problems can be formulated as 
constrained stochastic optimization problems, 
with the constraints originating mainly 
from model assumptions that define a trade-off between 
complexity and performance.
This trade-off is closely related to 
over-fitting, generalization capacity, and robustness to 
noise and adversarial attacks, 
and
depends on both the structure and complexity of the model, 
as well as the properties
of the optimization methods used.
We develop an online prototype-based learning algorithm 
based on annealing optimization
that is formulated as an online gradient-free stochastic approximation algorithm.
The learning model can be viewed as an interpretable and 
progressively growing competitive-learning neural network model
to be used for supervised, unsupervised, and reinforcement learning.
The annealing nature of the algorithm contributes to
minimal hyper-parameter tuning requirements, 
poor local minima prevention, and 
robustness with respect to the initial conditions.
At the same time, it provides online control over the performance-complexity trade-off
by progressively increasing the complexity of the learning model as needed, 
through an intuitive bifurcation phenomenon.
Finally, the use of stochastic approximation enables the study 
of the convergence of the learning algorithm through 
mathematical tools from dynamical systems and control, 
and allows for its integration with reinforcement learning algorithms,
constructing an adaptive state-action aggregation scheme.
%
%
\end{abstract}


\section{Introduction}
\label{Sec:Introduction}

Learning from data samples has proven to be an important component 
in the advancement of diverse fields, including artificial intelligence,
computational physics, biological sciences, 
communication frameworks, and cyber-physical control systems. 
While virtually all learning problems can be formulated as 
constrained stochastic optimization problems, 
the optimization methods can be intractable,  
with the constraints originating mainly 
from model assumptions and defining a trade-off between 
complexity and performance \cite{bennett2006interplay}.
For this reason, designing models of appropriate structure, 
and optimization methods with particular properties, 
has been the cornerstone
of machine learning algorithms.

Currently, deep learning methods
dominate the field of machine learning
owing to their experimental performance in numerous applications
\cite{lecun2015deep}.
%
However, they typically consist of overly complex models 
of a great many parameters,
which comes in the expense of time, energy, data, memory, 
and computational resources \cite{thompson2020computational,strubell2019energy}. 
Furthermore, they are inherently uninterpretable 
and vulnerable to small perturbations
\cite{szegedy2013intriguing,carlini2017towards},
which has led to an emerging hesitation in 
their usage outside common benchmark datasets 
and real-life or security critical applications 
\cite{sehwag2019analyzing}.

In this work, we introduce a learning model designed to 
alleviate these limitations to meet the needs of 
applications in which computational resources are limited and robustness
and interpretability are prioritized.
To that end, the learning model should create a meaningful representation, 
should be updated recursively (and even in real time)
with easy-to-implement updates,
and its complexity should be appropriately and progressively adjusted 
to offer online control over the trade-off between 
model complexity and performance.
This trade-off is also closely related to 
over-fitting, generalization capacity, and robustness to 
input perturbations and adversarial attacks 
\cite{xu2012robustness}.
This is further reinforced by recent studies 
revealing that existing flaws in the current 
benchmark datasets may have inflated the need for overly
complex models \cite{northcutt2021pervasive}, and that
over-fitting to adversarial training examples 
may actually hurt generalization \cite{raghunathan2019adversarial}. 

We focus on prototype-based models 
\cite{mavridis2020convergence,mavridis2022online,biehl2016prototype}, 
which are iterative, consistent \cite{mavridis2020convergence}, interpretable, robust \cite{mavridis2022risk}, and topology-preserving
competitive-learning neural networks \cite{martin_topologyPreservationInSOM_2008},
sparse in the sense of memory complexity, fast to train and evaluate,
and have recently 
shown impressive robustness 
against adversarial attacks,
suggesting suitability in security critical 
applications \cite{saralajew2019robustness}. 
They use a set of representatives (typically called prototypes, or codevectors)
to partition the input space in an optimal way according to an 
appropriately defined objective function.
This is an intuitive approach 
which parallels similar concepts from cognitive psychology
and neuroscience.
We approximate the global minima 
of the objective function 
by solving a sequence of optimization sub-problems 
that make use of entropy regularization at different levels.
This is a deterministic annealing approach
\cite{mavridis2022online,rose1998deterministic} 
that (a) adjusts the number of prototypes/neurons (which defines the complexity of the model) 
as needed through an intuitive bifurcation phenomenon,
(b) offers robustness with respect to the initial conditions, and
(c) generalizes the proximity measures used to quantify the similarity between
two vectors in the data space beyond convex metrics.
In addition, the annealing nature of the algorithm contributes to (but does not guarantee)
avoiding poor local minima, 
requires minimal 
hyper-parameter tuning, and 
allows online control over the performance-complexity trade-off.
%

Although deterministic annealing approaches have been known for a while 
\cite{rose1998deterministic}, 
an online optimization method for such architectures is an important development,
similar to the introduction of a greedy online training algorithm for 
a network of restricted Boltzmann machines 
that gave rise to one of the first effective deep learning algorithms
\cite{hinton2006fast}.
We develop an online training rule based on 
a stochastic approximation algorithm
\cite{robbins1951stochastic,borkar2009stochastic}
and show that it is also gradient-free, provided that the proximity measure 
used belongs to the family of Bregman divergences:
information-theoretic dissimilarity measures
that play an important role in learning applications 
and include the widely used Euclidean distance and 
Kullback-Leibler divergence
\cite{banerjee2005clustering,villmann_onlineDLVQmath_2010}.
While stochastic approximation 
offers an online, adaptive, 
and computationally inexpensive optimization framework,
it is also strongly connected to dynamical systems.
This enables the study 
of the convergence of the learning algorithm through 
mathematical tools from dynamical systems and control 
\cite{borkar2009stochastic}.
We take advantage of this property to prove the consistency of the 
proposed learning algorithm as a density estimator (unsupervised learning),
and as a classification rule (supervised learning).
Moreover, we make use of the theory of two-timescale 
stochastic approximation to show that the proposed learning
algorithm can be used as an adaptive aggregation 
scheme in reinforcement learning settings with:
(a) a fast component that executes a temporal-difference learning algorithm, 
and (b) a slow component for adaptive aggregation of the state-action space.
Finally, we illustrate the properties and evaluate the performance 
of the proposed learning algorithm in multiple experiments.

In particular, we start by dedicating Section \ref{app:sa}
to reviewing the theory of stochastic approximation as an optimization approach for 
learning algorithms,
giving emphasis to its connection to dynamical systems.
This concise background is targeted towards broader audience and aims to motivate
the generalization of training updates in learning algorithms.
We follow with Section \ref{Sec:ODA}, where  
we introduce the Online Deterministic Annealing (ODA) algorithm for unsupervised and
supervised learning and
study its convergence and practical implementation.
In Section \ref{Sec:RL} we show how ODA can be integrated with 
common reinforcement learning approaches, and in particular as an adaptive 
state-action aggregation algorithm that allows Q-learning to be applied to 
Markov decision processes with infinite-dimensional state and input spaces.
Finally, Section \ref{Sec:Results} illustrates experimental results, and 
Section \ref{Sec:Conclusion} concludes the paper.

\section{Stochastic Approximation: Learning with Dynamical Systems}
\label{app:sa}

In this section we briefly review the theory of stochastic approximation 
which is going to form the base for 
the convergence analysis of the proposed learning schemes.
We give emphasis to its connection to dynamical systems, and how 
this property can be particularly useful to optimization and 
machine learning algorithms.

\subsection{Stochastic Approximation and Dynamical Systems}
\label{sSec:SA}

Stochastic approximation, first introduced in 
\cite{robbins1951stochastic},
was originally conceived as a tool for statistical computation, 
%
and, since then, has become a central tool in a number of different disciplines, 
often times unbeknownst to the users, researchers and practitioners.
%
%
Stochastic approximation offers an online, adaptive, 
and computationally inexpensive optimization framework, 
properties that make it an ideal optimization method
for machine learning algorithms. 
As a result, many of the most widely used learning algorithms 
partially or entirely consist of stochastic approximation algorithms;
from stochastic gradient descent used in the back-propagation algorithm 
to train artificial neural networks
\cite{rumelhart1986learning,bottou1998online}, 
to the Q-learning algorithm 
used in reinforcement learning applications 
\cite{watkins1992q,tsitsiklis1994asynchronous}.
%

In addition to its connection with optimization and learning 
algorithms, 
stochastic approximation is also strongly connected to dynamical systems. 
A fact that is often overlooked is that almost any recursive numerical algorithm 
can be described by a discrete time dynamical system. 
In this sense, results about the behavior, 
e.g. stability and convergence properties, 
of discrete time dynamical systems can be applied to 
iterative optimization and learning algorithms. 
%
%
This connection is remarkably direct in stochastic approximation 
which allows the study of its convergence through the analysis of 
an ordinary differential equation, 
as illustrated in the following theorem, proven in 
\cite{borkar2009stochastic}:

\begin{theorem}[\cite{borkar2009stochastic}, Ch.2]
\label{thm:borkar}
	Almost surely, the sequence $\cbra{x_n}\in S\subseteq\mathbb{R}^d$ 
	generated by the following stochastic approximation scheme:
	\begin{align}
		x_{n+1} = x_n + \alpha(n) \sbra{h(x_n) + M_{n+1}},\ n \geq 0	
	\label{eq:sa}	
	\end{align}
	with prescribed $x_0$, 
	\textit{converges} to a (possibly sample path dependent)
	compact, connected, internally chain transitive, invariant set
	of the o.d.e:
	\begin{align}
		\dot{x}(t) = h\pbra{x(t)}, ~ t \geq 0, 	
	\label{eq:sa_ode}	
	\end{align}
	where $x:\mathbb{R}_+\rightarrow\mathbb{R}_d$ and $x(0) = x_0$, 
	provided the following assumptions hold:
	\begin{itemize}
	\setlength\itemsep{0em}
	\item[(A1)] The map $h:\mathbb{R}^d \rightarrow \mathbb{R}^d$ is Lipschitz
		in $S$,	i.e., $\exists L$ with $0 < L < \infty$ such that
		$\norm{h(x)-h(y)} \leq L\norm{x-y}, ~ x,y \in S$,
	\item[(A2)] The stepsizes $\cbra{\alpha(n) \in \mathbb{R}_{++}, ~ n \geq 0}$
	satisfy
		$ \sum_n \alpha(n) = \infty$, and $\sum_n \alpha^2(n) < \infty$	,
	\item[(A3)] $\cbra{M_n}$ is a martingale difference sequence 
		with respect to the increasing family of $\sigma$-fields
		$ \mathcal{F}_n := \sigma \pbra{ x_m, M_m,~ m \leq n }$, ${n \geq 0}$,
		i.e., $\E{M_{n+1}|\mathcal{F}_n} = 0 ~ a.s.$, for all $n \geq 0$,
		and, furthermore, $\cbra{M_{n}}$ are square-integrable with 
		$ \E{\norm{M_{n+1}}^2|\mathcal{F}_n} \leq K \pbra{ 1 + \norm{x_n}^2 }, 
		~ a.s.$, where $n \geq 0 $ for some $K >0$,
	\item[(A4)] The iterates $\cbra{x_n}$ remain bounded a.s., i.e.,
		${ \sup_n \norm{x_n} < \infty}$ $ a.s.$
	\end{itemize}
\end{theorem}

Intuitively, the stochastic process (\ref{eq:sa}) can be seen
as a noisy discretization 
(also known as Euler scheme in numerical analysis literature)
of (\ref{eq:sa_ode}).
As an immediate result, the following corollary also holds:
\begin{corollary}[\cite{borkar2009stochastic}]
If the only internally chain transitive invariant sets for
(\ref{eq:sa_ode}) are isolated equilibrium points,
then, almost surely, $\cbra{x_n}$ converges to a, 
possibly sample dependent, equilibrium point of (\ref{eq:sa_ode}).  
\label{crl:sa_equillibria}
\end{corollary}

Given the conditions of Theorem \ref{thm:borkar} 
and using standard Lyapunov arguments, the following corollary, 
regarding distributed, asynchronous implementation of the algorithm, also holds:
\begin{corollary}[\cite{borkar2009stochastic}, Ch. 7]	
Suppose there exists a continuously differentiable function $J$, 
such that $h(x)=-\nabla J(x)$ (or $h(x)=F(x)-x$). 
Define $Y_n\subseteq \cbra{1,\ldots,d}$ to be 
the subset of components of $x_n$ that are updated at time $n$, and 
$v(i,n):= \sum_{m=0}^n \mathds{1}_{\sbra{i\in Y_m}}$ to be 
the number of times the $i$-th component $x_n^{(i)}$ has been updated 
up until time $n$.
Then, almost surely, the sequence $\cbra{x_n}$ generated by
\begin{align}
	x_{n+1}^{(i)} = x_n^{(i)} + \alpha(v(i,n)) \mathds{1}_{[i\in Y_n]} 
		\sbra{h^{(i)}(x_n) + M_{n+1}^{(i)}}
\label{eq:sa_as}	
\end{align}
where $i\in\cbra{1,\ldots,d}$, and $n \geq 0$, converge to 
the invariant set $H:=\cbra{x: \nabla J(x)=0}$ 
(or $H:=\cbra{x: F(x)=x}$), 
provided that each component $(i)$ is updated infinitely often, i.e.
\begin{align*}
\lim_{n\rightarrow\infty} \inf \frac{v(i,n)}{n} > 0.
\end{align*}	 	   
\label{crl:sa_as}
\end{corollary}

Corollaries \ref{crl:sa_equillibria} and \ref{crl:sa_as}, reveal the 
connection of the stochastic approximation algorithms with 
iterative approximation and optimization algorithms, including 
two notable special cases: 
stochastic gradient descent, and
the Q-learning algorithm
%
These special cases of stochastic approximation,
are discussed in more detail in what follows.

\subsection{Stochastic Gradient Descent}
\label{sSec:SGD}

%
Stochastic gradient descent is an iterative stochastic optimization
method that tries to solve the problem of  
minimizing the cost:
\begin{equation}
\min_\theta \E{J(X,\theta)}
\end{equation}
where $X:\Omega\rightarrow H$ is a random variable defined in 
the probability space $(\Omega,\mathcal{F},\mathbb{P})$, and 
$H$ is a Hilbert space.
%
%
%
The update
\begin{equation}
\theta_{n+1} = \theta_n - \alpha_n \nabla_\theta J(x_n,\theta_n)
\label{eq:sgd}
\end{equation}
is used
to bypass the estimation of $\Eh{J(x,\theta_n)}= \frac{1}{n}\sum_{i=1}^n J(x_i,\theta_n)$ 
which can be expensive or infeasible. 
%
%
This is a special 
case of a stochastic approximation update.
Observe that, 
under the condition 
$\nabla_\theta \E{J(x,\theta_n)} = \E{\nabla_\theta J(x,\theta_n)}$, (\ref{eq:sgd}) can be written as:
\begin{equation}
\theta_{n+1} = \theta_n + \alpha_n \sbra{ 
-\nabla_\theta \E{J} + (\E{\nabla_\theta J} - \nabla_\theta J)}
\end{equation}
where $h(\theta)=-\nabla_\theta \E{J(X,\theta)}$ 
is a Lipschitz continuous function, 
and $M_n=\E{\nabla_\theta J(X,\theta)} - \nabla_\theta J(x_n,\theta)$ is a martingale difference
sequence, since the data samples $x_n$ are assumed independent realizations
of the random variable $X$.
Therefore by Theorem \ref{thm:borkar} 
and Corollary \ref{crl:sa_equillibria},
as long as $ \sum_n \alpha(n) = \infty$, and $\sum_n \alpha^2(n) < \infty$,
and $\theta_n$ remain bounded a.s., stochastic gradient descent will 
converge to a possibly path dependent equilibrium of 
$\dot \theta = -\nabla_\theta \E{J(X,\theta)}$, i.e., 
in a minimizer of $\E{J(X,\theta)}$.

\subsection{Q-learning}
\label{sSec:Qlearning}

As a second example, the Q-learning algorithm, widely used in reinforcement
learning, is again a special case of a stochastic approximation algorithm
\cite{borkar2000ode}.
Consider a discrete-time Markov Decision Process (MDP) $(\mathcal{X}, \mathcal{U}, \mathcal{P}, C)$
with:
%
\begin{itemize}
\item $\mathcal{X}$ being the state space,
\item $\mathcal{U}$ being the action (control) space,
\item $\mathcal{P}:(x,u,x') \mapsto \mathbb{P}\sbra{x'|x,u}$ being
	the transition probabilities associated with  
	a stochastic state transition function 	 
	${f:(x,u)\mapsto x'}$, and
\item $C:\mathcal{X}\times\mathcal{U}\rightarrow\mathbb{R}_+$,
	being the immediate cost function, assumed deterministic.
\end{itemize}
%
Reinforcement Learning (RL) examines the problem of learning 
a control policy $u:=(u_0, u_1, \ldots)$ 
that solves the discounted infinite-horizon optimal control problem
\begin{align*}
\min_u \E{ \sum_{l=0}^\infty \gamma^l C(x_l,u_l)}
\end{align*}
where $\gamma\in(0,1]$.
Define the value function $V^u$ of a policy $u$ as 
\begin{align*}
V^u(x_k) :&= \E{ \sum_{l=k}^\infty \gamma^{l-k} C(x_l,u_l)} \\
	&= C(x_k,u_k) + \gamma \E{V^u(x_{k+1})\mid x_k} \\
	&= Q^u(x_k,u_k)
\end{align*}
where $Q^u$ represents the quality function of a policy $u$, i.e.
the expected return for taking action $u_k$ at time $k$ and state 
$x_k$, and thereafter following policy $u$. 
As a result of Bellman's principle, we get the (discrete-time) 
Hamilton-Jacobi-Bellman (HJB) equation 
\begin{equation}
\begin{aligned}
V^*(x_k):&= 
	\min_{u}~ \E{ \sum_{l=k}^\infty \gamma^{l-k} C(x_l,u_l) } \\
	&\overset{(HJB)}{=}	\min_{u} \cbra{~ C(x_k,u_k) + \gamma \E{V^*(x_{k+1})\mid x_k} ~} \\
	&= \min_{u_k} Q^*(x_k,u_k)
\end{aligned}
\label{eq:HJB}
\end{equation}
where $V^*:=V^{u^*}$ and  $Q^*:=Q^{u^*}$ represent the optimal value and
$Q$ functions, respectively.
Reinforcement learning algorithms consist mainly of
temporal-difference learning algorithms 
that try to approximate a solution to (\ref{eq:HJB})
using iterative optimization methods.
The optimization is performed over a finite set of parameters 
which are used to describe the value (or Q) function.
These parameters typically correspond to a parametric model 
(e.g. a neural network) used for function approximation, or
to the different values of the vector $V(\mathcal{X})$ (or $Q(\mathcal{X},\mathcal{U})$),
in which case $\mathcal{X}$ and $\mathcal{U}$ are assumed finite 
either by definition or as a result of discretization.
Assuming that the state and action spaces
$\mathcal{X}$ and $\mathcal{U}$ are finite, 
a widely used approach is  
the $Q$-learning algorithm:
\begin{align*}
Q_{j+1}(x,u') &= Q_{j}(x,u') + \alpha_j [ C(x,u')\\
 &~~~~~~~~~~~~~ + \gamma \min_{u} Q_{j}(x',u) - Q_{j}(x,u') ]
\end{align*}
which is a stochastic approximation algorithm \cite{borkar2000ode}:
\begin{equation}
\begin{aligned}
& Q_{j+1}(x,u') = Q_{j}(x,u') + \alpha_j [ \\
&~~ \pbra{ C(x,u') + \gamma  \E{\min_{u} Q_{j}(x,u)}  - Q_{j}(x,u') } \\
&~~ + \gamma \pbra{ \min_{u} Q_{j}(x',u) - \E{\min_{u} Q_{j}(x,u)} } ]
\end{aligned}
\end{equation}
with $h(Q(x,u'))=C(x,u') + \gamma  \E{\min_{u} Q_{j}(x,u)}  - Q_{j}(x,u')$,
and $M_{j}=\gamma \pbra{ \min_{u} Q_{j}(x',u) - \E{\min_{u} Q_{j}(x,u)} }$
is a martingale difference sequence.
As a result, under the conditions of Theorem \ref{thm:borkar},
the Q-learning algorithm converges to the global
equilibrium of $\dot Q = F(Q)-Q$, with 
$F(Q(x,u')) = C(x,u') + \gamma  \E{\min_{u} Q_{j}(x,u)}$, i.e., 
to the stationary point $Q(x,u) = C(x,u) + \gamma  \E{\min_{u} Q_{j}(x,u)}$, 
which solves the Hamilton-Jacobi-Bellman equation.

\subsection{Dynamics and Control for Learning}

It follows from the above that 
stochastic approximation algorithms define a family of iterative 
approximation and optimization algorithms that can be used, 
among others, for machine learning applications. 
%
Their strong connection to dynamical systems (see Theorem \ref{thm:borkar}),
can give rise to the study of 
learning algorithms and representations through 
systems-theoretic mathematics, connecting machine learning with
stochastic optimization, 
adaptive control and dynamical systems, 
which   
%
can lead to new developments in the field of machine learning.
As a first example, notice that (\ref{eq:sa}) defines an iterative 
algorithm that can be used for stochastic optimization and does not 
necessarily depend on the gradient of a cost function.
As will be
shown in Section \ref{Sec:ODA},
this can lead to gradient-free learning algorithms that can alleviate 
common problems such as that of the vanishing gradients.
%

In addition, the developed mathematical theory of dynamical systems 
can be utilized to construct and study learning algorithms 
that run at the same time but at different timescales.
In particular, the theory of the O.D.E. method 
for stochastic approximation in two timescales as detailed in
\cite{borkar2009stochastic} is summarized in the following theorem:

\begin{theorem}[Ch. 6 of \cite{borkar1997stochastic}]
Consider the sequences $\cbra{x_n}\in S\subseteq\mathbb{R}^d$ and 
$\cbra{y_n}\in \Sigma\subseteq\mathbb{R}^k$,
generated by the iterative stochastic approximation schemes:
\begin{align}
x_{n+1} = x_n + \alpha(n) \sbra{f(x_n,y_n) + M_{n+1}^{(x)}} 
\label{eq:sa_timescales_x}	\\
y_{n+1} = y_n + \beta(n) \sbra{g(x_n,y_n) + M_{n+1}^{(y)}}
\label{eq:sa_timescales_y}
\end{align}

for $n \geq 0$ and $M_n^{(x)}$, $M_n^{(y)}$ martingale difference sequences, 
and assume that
$\sum_n \alpha(n) = \sum_n \beta(n) = \infty$, 
$\sum_n ( \alpha^2(n)+\beta^2(n) ) <\infty$, and 
$\nicefrac{\beta(n)}{\alpha(n)}\rightarrow 0$,
%
%
with the last condition implying that the iterations for $\cbra{y_n}$
run on a slower timescale than those for $\cbra{x_n}$. 
If the equation 
\begin{equation*}
\dot x (t) = f(x(t),y),\ x(0)=x_0
\end{equation*}
has an asymptotically stable equilibrium $\lambda(y)$ 
for fixed $y$ and some Lipschitz mapping $\lambda$, and the equation 
\begin{equation*}
\dot y (t) = g(\lambda(y(t)),y(t)),\ y(0)=y_0
\end{equation*}
has an asymptotically stable equilibrium $y^*$, 
then, almost surely, $(x_n,y_n)$ converges to $(\lambda(y^*),y^*)$.
\label{thm:borkar_timescales}
\end{theorem}
This result allows for two learning algorithms,
that may depend on each other,
to run online at the same time, but at different timescales.
%
As will be shown in Section \ref{Sec:RL}, 
a two-timescale stochastic approximation algorithm can be used 
for reinforcement learning with: 
(a) a fast component that executes a Q-learning algorithm, 
and (b) a slow component, that adaptively partitions
the state-action space according to an appropriately defined 
dissimilarity measure.

\section{Online Deterministic Annealing for Unsupervised and Supervised Learning}
\label{Sec:ODA}

To formulate the mathematics of a prototype-based 
learning model that progressively grows in size, 
it is convenient to start our analysis with the case of unsupervised learning, i.e., 
clustering and density estimation, 
and then show how these results generalize in the supervised case, i.e., 
in classification and regression,
as well.
In the context of unsupervised learning, 
the observations (data) 
are represented by a random variable $X: \Omega \rightarrow S$ defined in a 
probability space $\pbra{\Omega, \mathcal{F}, \mathbb{P}}$, 
where $S\subseteq \mathbb{R}^d$ is the observation space (data space).
The goal of prototype-based learning is to define a 
similarity measure $d:S\rightarrow ri(S)$ 
(where $ri(S)$ represents the relative interior of $S$) and a
set of $K$ prototypes $\mu := \cbra{\mu_i}_{i = 1}^K$, $\mu_i \in ri(S)$, 
on the data space
such that an 
average distortion measure is minimized, i.e., 
\begin{equation}
    \min_{\mu} ~ J(\mu) := \E{\min_i d(X,\mu_i)}
    \label{eq:VQ}
\end{equation}
Here the similarity measure as well as the number of prototypes $K$
are predefined designer parameters. 
This process is equivalent to finding the most suitable model out of a set of $K$
local constant models, and results in a piecewise-constant approximation of 
the data space $S$.
This representation has been used for clustering in vector quantization 
applications \cite{mavridis2020convergence, Kohonen1995}, 
and, in the limit $K\rightarrow\infty$, can be used for density estimation.

To construct a learning algorithm that progressively increases the number 
of prototypes $K$ as needed 
according to different ``levels of detail'' (to be defined shortly)
we will define a probability space over an infinite number of candidate models, 
and constraint their distribution using the maximum-entropy principle 
at different levels.
As we will show, solving a sequence of optimization problems 
parameterized by a single parameter $T$
will result in a series of model distributions 
with a finite number of $K(T)$ values with non-zero probability, i.e., 
this process results in a finite number of ``effective codevectors'' $K(T)$
that depends on a ``temperature parameter'' $T$.

First we need to
adopt a probabilistic approach for (\ref{eq:VQ}), in which 
a quantizer $Q:S \rightarrow ri(S)$ 
is defined as a 
discrete random variable with domain $\mu := \cbra{\mu_i}_{i = 1}^K$,
such that (\ref{eq:VQ}) becomes
\begin{equation}
\begin{aligned}
\min_{\mu} ~ D(\mu) &:= \E{d\pbra{X,Q}} 
\\& 
=\E{\E{d(X,Q)|X}} 
\\&
=\int p(x) \sum_i p(\mu_i|x) d_\phi(x,\mu_i) ~dx
\end{aligned}
\label{eq:softVQ}
\end{equation}	
This is now a problem of finding the locations $\cbra{\mu_i}$ and the 
association probabilities
$\cbra{p(\mu_i|x)}:=\cbra{p(Q=\mu_i|X=x)}$.
Notice that this is a more general problem than that of (\ref{eq:VQ}),
where it is subtly assumed that 
$p(\mu_i|x)=\mathds{1}_{\sbra{x \in S_i}}$, where 
$S_i = \cbra{x \in S: i = \argmin\limits_{j = 
		1,\ldots,K} ~ d(x,\mu_{j})},\ i =1,\ldots,K$,
and $\cbra{S_i}$ defines a Voronoi partition.

Now, we make the assumption that 
we have access to an infinite number of possible models, i.e., 
that the quantizer $Q$ is a discrete random variable
over a countably infinite set $\mu := \cbra{\mu_i}$. 
Instead of choosing $K$ a priori, we will constraint 
the model distribution at different levels by maximizing the entropy:

\begin{equation}
\begin{aligned}
H(\mu) &:= \E{-\log P(X,Q)} 
\\&
=H(X) + H(Q|X)
\\&
=H(X) - \int p(x) \sum_i p(\mu_i|x) \log p(\mu_i|x) ~dx
\end{aligned}    
\end{equation}
This is essentially a realization of the 
Jaynes's maximum entropy principle \cite{jaynes1957information}.
%
We formulate the resulting multi-objective optimization as the minimization of the Lagrangian
\begin{equation}
\min_\mu F(\mu) := D(\mu) - T H(\mu)
\label{eq:F}
\end{equation}
where $T\in[0,\infty)$ acts as a Lagrange multiplier,
and, as we will show, can be seen as a temperature coefficient 
in an annealing process \cite{mavridis2022online,rose1998deterministic}.

\begin{remark}
Alternatively, (\ref{eq:F}) can be formulated as
\begin{equation}
\min_\mu F(\mu) := (1-\lambda) D(\mu) - \lambda H(\mu)
\label{eq:Flambda}
\end{equation}
where $\lambda\in\sbra{0,1}$, with
\begin{equation}
 T:=\frac{1-\lambda}{\lambda}  ,\quad \lambda\in\sbra{0,1}
 \label{eq:Tlambda}
\end{equation}
representing the corresponding temperature coefficient.
This is a mathematically equivalent formulation that, as will be discussed, 
can yield major benefits in the algorithmic implementation.
\end{remark}

Equation (\ref{eq:F}) represents the scalarization method for trade-off 
analysis between two performance metrics, one related to performance, and 
one to generalization.
The entropy $H$, 
acts as a regularization term, and is given progressively less weight 
as $T$ decreases.
For large values of $T\rightarrow \infty$ we maximize the entropy.
As we will show, this results 
in a unique effective codevector that represents the entire data space.
As $T$ is lowered, we essentially transition 
from one solution of the multi-objective optimization 
(a Pareto point when the objectives are convex) to another in a naturally occurring direction
that resembles an annealing process.

In this sense, the value of $T$ defines the ``level of detail'' of the dataset 
that is allowed to be seen by the maximum-entropy constraint. 
As we will show, when certain critical values of $T$ are reached, 
a bifurcation phenomenon occurs, according to which, the the number of non-zero values 
$K(T)$ of the 
model distribution increases, indicating that  
the solution to the optimization problem of minimizing $F(T)$ requires more ``effective codevectors'' $K(T)$.

\begin{remark}
We note that this concept is similar to
convex relaxation of hierarchical clustering, which
results in a family of objective functions with a natural geometric interpretation 
\cite{hocking2011clusterpath}. However, as we will show, the proposed approach does not 
make any relaxation assumptions, uses entropy as a naturally occurring regularization term, 
and allows for the development of a gradient-free training rule based on 
stochastic approximation.
This will result in a learning algorithm that can be integrated with reinforcement learning 
approaches.
\end{remark}

\subsection{Solving the Optimization Problem}

As in the case of standard vector quantization algorithms, 
we will minimize $F$ 
by successively minimizing it first respect to the 
association probabilities $\cbra{p(\mu_i|x)}$, 
and then with respect to the codevector locations $\mu$.

The following lemma provides the solution of 
minimizing $F$ with respect to the association probabilities $p(\mu_i|x)$: 
\begin{lemma}
The solution of the optimization problem
\begin{equation}
\begin{aligned}
F^*(\mu) &:= \min_{\cbra{p(\mu_i|x)}} F(\mu)
\\
\text{s.t.} & \sum_i p(\mu_i|x) = 1 
\end{aligned}
\label{eq:Fstar}
\end{equation}
is given by the Gibbs distributions
\begin{equation}
p^*(\mu_i|x) = \frac{e^{-\frac{d(x,\mu_i)}{T}}}
			{\sum_j e^{-\frac{d(x,\mu_j)}{T}}} ,~ \forall x\in S
\label{eq:gibbs}
\end{equation}
\end{lemma}
\begin{proof}
We form the Lagrangian:
\begin{equation}
\begin{aligned}
    \mathcal L_f(\cbra{p(\mu_i|x)},\nu)
    &:= D(\mu) -  T H(\mu) + \nu \pbra{\sum_i p(\mu_i|x) - 1} 
\end{aligned}    
\end{equation}
Taking $\frac{\partial \mathcal L}{\partial p(\mu|x) } = 0 $ yields:
\begin{equation}
\begin{aligned}
    & d(x,\mu_i) + T (1+\log p(\mu_i|x)) + \nu = 0 \\
    & \implies \log p(\mu_i|x) = - \frac{1}{T} d(x,\mu_i)
            - \pbra{ 1 + \frac \nu T } \\
    & \implies p(\mu_i|x) =  \frac{e^{- \frac{d(x,\mu_i)}{T}}}{e^{1 + \frac \nu T}}
\end{aligned}
\end{equation}
Finally, from the condition $\sum_i p(\mu_i|x) = 1$, it follows that 
\begin{equation}
    e^{1 + \frac \nu \lambda} = \sum_i e^{- \frac{d(x,\mu_i)}{T} }
\end{equation}
which completes the proof.
\end{proof}

In order to minimize $F^*(\mu)$ with respect to the codevector locations $\mu$ 
we set the gradients to zero 
\begin{equation}
\begin{aligned}
&\frac d {d\mu} F^*(\mu) = 0 
\implies 
\frac d {d\mu} \pbra{ D(\mu) - T H(\mu) } = 0 
\\&
\implies
\int p(x) \sum_i \frac d {d\mu} \pbra{p^*(\mu_i|x) d_\phi(x,\mu_i)} 
    \\&\quad\quad\quad + T \frac d {d\mu} \pbra{p^*(\mu_i|x) \log p^*(\mu_i|x)} ~dx = 0
\\&
\implies
\sum_i \int p(x) p^*(\mu_i|x) \frac d {d\mu_i} d(x,\mu_i) ~dx = 0
\end{aligned}
\label{eq:M}
\end{equation}
where we have used (\ref{eq:gibbs}), direct differentiation, and the fact that 
$\sum_i \frac d {d\mu} p^*(\mu_i|x) = \frac d {d\mu} \sum_i p^*(\mu_i|x) = 0$.
In the next section, we show that 
(\ref{eq:M}) has a
closed form solution if the dissimilarity measure $d$ 
belongs to the family of Bregman divergences.

\subsection{Bregman Divergences as Dissimilarity Measures}

Prototype-based algorithms 
rely on measuring the proximity between different vector representations.
In most cases the Euclidean distance or another convex metric is used,
but this can be generalized to alternative dissimilarity measures inspired by 
information theory and statistical analysis, such as the Bregman 
divergences \cite{banerjee2005clustering}:
%
%
\begin{definition}[Bregman Divergence]
	Let $ \phi: H \rightarrow \mathbb{R}$, 
	be a strictly convex function defined on 
	a vector space $H$ such that $\phi$  
	is twice F-differentiable on $H$. 
	The Bregman divergence 
	$d_{\phi}:H \times H \rightarrow \left[0,\infty\right)$
	is defined as:
	\begin{align*}
		d_{\phi} \pbra{x, \mu} = \phi \pbra{x} - \phi \pbra{\mu} 
							- \pder{\phi}{\mu} \pbra{\mu} \pbra{x-\mu},
	\end{align*}
	where $x,\mu\in H$, and the continuous linear map 
	$\pder{\phi}{\mu} \pbra{\mu}: H \rightarrow \mathbb{R}$ 
	is the Fr\'echet derivative of $\phi$ at $\mu$.
	\label{def:BregmanD}
\end{definition}
In this work, we will concentrate on nonempty, compact convex sets 
$S\subseteq \mathbb{R}^d$ 
so that the derivative of $d_\phi$ with respect to the second argument 
can be written as
{\small
\begin{align}
\pder{d_{\phi}}{\mu}(x,\mu) 
&= \pder{\phi(x)}{\mu} - \pder{\phi(\mu)}{\mu} 
- \pder{^2 \phi(\mu)}{\mu^2}(x-\mu) + \pder{\phi(\mu)}{\mu} \\
&= - \pder{^2 \phi(\mu)}{\mu^2}(x-\mu) 
= - \abra{\nabla^2 \phi(\mu),(x-\mu)}	
\label{eq:dd_phi}
\end{align}
}
where $x,\mu\in S$, $\pder{}{\mu}$ represents differentiation
with respect to the second argument of $d_{\phi}$, and 
$\nabla^2 \phi(\mu)$ represents the Hessian matrix of $\phi$ at $\mu$.
%

%
%


\begin{example}
As a first example, $\phi(x) = \abra{x,x},\ x\in\mathbb{R}^d$,
yields the squared Euclidean distance 
$d_\phi(x,\mu) = \|x-\mu\|^2$. 
\end{example}
\begin{example}
A second interesting Bregman divergence that shows the connection 
to information theory, is the generalized I-divergence 
which results from 
$\phi(x) = \abra{x,\log x},\ x\in\mathbb{R}_{++}^d$
such that  
$$d_\phi(x,y) = \abra{x,\log x - \log \mu}
	- \abra{\mathds{1}, x - \mu,}$$
where $\mathds{1}\in\mathbb{R}^d$ is the vector of ones.
It is easy to see that $d_\phi(x)$ reduces to the Kullback-Leibler divergence if 
$\abra{\mathds{1}, x} =1$.
\end{example}
%

The family of Bregman divergences provides proximity measures
that have been shown to enhance the performance of a learning algorithm
\cite{villmann_onlineDLVQmath_2010}.
%
%
There is also a deeper connection of Bregman divergences to 
prototype-based learning algorithms \cite{banerjee2005clustering}.
In the next theorem, we show that
we can have analytical solution to the last optimization step 
(\ref{eq:M}) in a convenient centroid form, if $d$ is a Bregman divergence.

\begin{theorem}
The optimization problem
\begin{equation}
    \min_\mu F^*(\mu)
    \label{eq:minFstar}
\end{equation}
where $F^*(\mu)$ is defined in (\ref{eq:Fstar}) is solved by the 
codevector locations $\mu$ given by  
\begin{equation}
\mu_i^* = \E{X|\mu_i} = \frac{\int x p(x) p^*(\mu_i|x) ~dx}{p^*(\mu_i)}
\label{eq:mu_star}
\end{equation}
if $d:=d_\phi$ is a Bregman divergence for some function 
$\phi$ that satisfies Definition \ref{def:BregmanD}.
\label{thm:bregman_in_DA}
\end{theorem}
\begin{proof}
Given (\ref{eq:dd_phi}), (\ref{eq:M}) becomes
\begin{equation}
\int (x-\mu_i) p(x) p^*(\mu_i|x) ~dx= 0
\end{equation}
which is equivalent to (\ref{eq:mu_star}) since 
$\int p(x) p^*(\mu_i|x) ~dx = p^*(\mu_i)$.
\end{proof}
%

\subsection{Bifurcation and The Number of Clusters}
\label{sSec:bifurcation}

So far, we have assumed a countably infinite set of codevectors.
In this section we will show that the distribution of the quantizer $Q$
is actually discrete and takes values from a finite set of $K(T)$ codevectors 
which we call ``effective codevectors''.
Both the number and the locations of the codevectors will depend on the value of
the temperature parameter $T$.
These effective codevectors are the only parameters that an algorithmic implementation will 
need to store in memory.

First, notice that as $T\rightarrow\infty$, equation (\ref{eq:gibbs}) yields
uniform association probabilities $p(\mu_i|x)=p(\mu_j|x),\ \forall i,j, \forall x$.
As a result of (\ref{eq:mu_star}), all codevectors are located at the same point:
\begin{align*}
\mu_i = \E{X},\ \forall i
\end{align*}
which means that there is one unique effective codevector given by $\E{X}$.

As $T$ is lowered below a critical value, a bifurcation phenomenon occurs, 
when the number of effective codevectors increases, 
which is a physical analogy with chemical annealing processes.
Mathematically, it occurs when the existing solution $\mu^*$ given by (\ref{eq:mu_star}) 
is no longer the minimum of the free energy $F^*$,
as the temperature $T$ crosses a critical value.
Following principles from variational calculus, 
we can rewrite the necessary condition for optimality (\ref{eq:M}) as
\begin{equation}
    \frac{d}{d\epsilon} F^*(\mu+\epsilon \psi)|_{\epsilon=0} = 0
\end{equation}
with the second order condition being 
\begin{equation}
    \frac{d^2}{d\epsilon^2} F^*(\cbra{\mu+\epsilon \psi})|_{\epsilon=0} \geq 0
    \label{eq:soc}
\end{equation}
for all choices of finite perturbations $\cbra{\psi}$.
Here we will denote by $\cbra{y := \mu + \epsilon \psi}$ a perturbed codebook, 
where $\psi$ are perturbation vectors applied to the codevectors $\mu$, and 
$\epsilon\geq 0$ is used to scale the magnitude of the perturbation. 
Bifurcation occurs when equality is achieved in (\ref{eq:soc})
and hence the minimum is no longer stable%
\footnote{For simplicity we ignore higher order derivatives, which should be checked for mathematical completeness, but which are of minimal practical importance. The result is a necessary condition for bifurcation.}.
These conditions are described in the following theorem:
\begin{theorem}
Bifurcation occurs under the following condition 
\begin{equation}
    \exists y_n \text{  s.t.  } p(y_n)>0 \text{  and  } \det\sbra{ I -  T \pder{^2 \phi(y_n)}{y_n^2} C_{x|y_n}} = 0
\end{equation}
where $C_{x|y_n} := \E{(x-y_n) (x-y_n)^\T|y_n}$.
\end{theorem}
\begin{proof}
From direct differentiation the optimality condition (\ref{eq:soc}) becomes
{\small
\begin{equation}
    \begin{aligned}
     &\sum_i p(y_i) \pder{^2 \phi(y_i)}{y_i^2} \psi^\T \sbra{ I -  T \pder{^2 \phi(y_i)}{y_i^2} C_{x|y_i} } \psi 
    \\&\quad
    + T \int p(x) \pbra{\sum_i p(y_i|x) \pder{^2 \phi(y_i)}{y_i^2}(x-y_i)^\T \psi }^2 dx
    = 0 
    \end{aligned}
    \label{eq:bifurcation:finalsoc}
\end{equation}
}
where 
{\small
\begin{equation}
    C_{x|y_i} := \E{(x-y_i) (x-y_i)^\T|y_i} = \int p(x|y_i) (x-y_i) (x-y_i)^\T dx
\end{equation}
}
The left-hand side of (\ref{eq:bifurcation:finalsoc}) 
is positive for all perturbations $\cbra{\psi}$
if and only if the first term is positive. 
To see that, notice that the second term of (\ref{eq:bifurcation:finalsoc})
is clearly non-negative. 
For the left-hand side to be non-positive, the first term needs to 
be non-positive as well, i.e., there should exist at least one codevector value, 
say $y_n$, such that $p(y_n)>0$ and 
$\sbra{ I -  T \pder{^2 \phi(y_n)}{y_n^2} C_{x|y_n}}\preceq 0$.
In this case, there always exist a perturbation vector $\cbra{y}$
such that $y = 0$, $\forall y\neq y_n$, and $\sum_{y=y_n} \psi =0 $, 
that vanishes the second term, i.e., 
$T \int p(x) \pbra{\sum_i p(y_i|x) \pder{^2 \phi(y_i)}{y_i^2}(x-y_i)^\T \psi }^2 dx= 0 $.
In other words we have shown that
\begin{equation}
\begin{aligned}
&\frac{d^2}{d\epsilon^2} F^*(y) > 0 \Leftrightarrow 
\\&\quad
\exists y_n \text{  s.t.  } p(y_n)>0 \text{  and  } \sbra{ I -  T \pder{^2 \phi(y_n)}{y_n^2} C_{x|y_n}}\succ 0
\end{aligned}
\end{equation}
%
which completes the proof. 
\end{proof}

Notice that loss of minimality also implies that the number of effective codevectors has changed, 
otherwise the minimum would be stable. 
In addition, we have showed that bifurcation depends on the temperature coefficient $T$ (and the choice of the Bregman divergence, through the function $\phi$) and 
occurs when 
\begin{equation}
    \frac{1}{T} = \pder{^2 \phi(y_n)}{y_n^2} \bar \nu
\end{equation}
where $\bar \nu$ is the largest eigenvalue of $C_{x|y_n}$.
As a result, the following corollary holds:
\begin{corollary}
The number of effective codevectors always 
remains bounded between two critical temperature values.    
\end{corollary}
In other words, an algorithmic implementation needs only
as many codevectors as the number of effective codevectors, which
depends only on changes of the temperature parameter below certain thresholds that
depend on the dataset at hand and the dissimilarity measure used.
As shown in Alg. \ref{alg:ODA}, 
we can detect the bifurcation points 
by introducing perturbing pairs of codevectors at each 
temperature level $T$.
In this way, the codevectors $\mu$ are doubled by inserting a perturbation of each $\mu_i$ in 
the set of effective codevectors.
The newly inserted codevectors will merge with their pair if 
a critical temperature has not been reached and separate otherwise.
For more details about the implementation of the algorithm the reader is referred to 
\cite{mavridis2022online}.

%

\subsection{The Online Learning Rule}

The conditional expectation $\E{X|\mu}$ in eq. (\ref{eq:mu_star})
can be approximated by the sample mean 
of the data points weighted 
by their association 
probabilities $p(\mu|x)$, i.e., 
$\Eh{X|\mu} = \frac{\sum x p(\mu|x)}{p(\mu)}$. 
This approach, however, defines an offline (batch) optimization algorithm 
and requires the entire dataset to be available a priori,
subtly assuming that it is possible to store 
and also quickly access the entire dataset at each iteration. 
This is rarely the case in practical applications and 
results to computationally costly iterations that are slow to converge.

We propose an Online Deterministic Annealing (ODA) algorithm, 
that dynamically updates 
its estimate of the effective codevectors with every observation.
This results in a significant reduction in complexity, 
that comes in two levels.
The first refers to huge reduction in memory complexity, 
since we bypass the need to store the entire dataset, 
as well as the association probabilities 
$\cbra{p(\mu_i|x),\ \forall x, i}$ that map each
data point in the dataset to each cluster.
The second level refers to the nature of the optimization iterations. 
In the online approach 
the optimization iterations increase in number
but become much faster, and practical convergence is often reached 
after a smaller number of observations.
%
%
%
%
%
To define an online training rule 
for the deterministic annealing framework,
we formulate a stochastic approximation algorithm 
to recursively estimate $\E{X|\mu}$ directly.
The following theorem provides a means towards constructing 
a gradient-free stochastic approximation training rule for 
the online deterministic annealing algorithm.

\begin{theorem}
Let $S$ a vector space, $\mu\in S$, and
$X: \Omega \rightarrow S$
be a random variable defined in a 
probability space $\pbra{\Omega, \mathcal{F}, \mathbb{P}}$.
Let $\cbra{x_n}$ be a sequence of independent realizations of $X$,
and $\cbra{\alpha(n)>0}$ a sequence of stepsizes such that
$ \sum_n \alpha(n) = \infty$, and $\sum_n \alpha^2(n) < \infty$.
Then the random variable $m_n = \nicefrac{\sigma_n}{\rho_n}$,
where $(\rho_n, \sigma_n)$ are 
sequences defined by
\begin{equation}
\begin{aligned}
\rho_{n+1} &= \rho_n + \alpha(n) \sbra{ p(\mu|x_n) - \rho_n} \\
\sigma_{n+1} &= \sigma_n + \alpha(n) \sbra{ x_n p(\mu|x_n) - \sigma_n},
\end{aligned}
\label{eq:rhosigma}
\end{equation}
converges to $\E{X|\mu}$ almost surely, i.e. 
$m_n\xrightarrow{a.s.} \E{X|\mu}$.
\label{thm:oda_sa}
\end{theorem}

\begin{proof}
We will use the facts that $p(\mu)=\E{p(\mu|x)}$ and 
$\E{\mathds{1}_{\sbra{\mu}}X} = \E{xp(\mu|x)}$.
The recursive equations (\ref{eq:rhosigma}) are 
stochastic approximation algorithms of the form:
\begin{equation}
\begin{aligned}
\rho_{n+1} &= \rho_n + \alpha(n)  
	[ (p(\mu) - \rho_n) + \\ 
	&\quad\quad\quad\quad\quad\quad\quad\quad
	(p(\mu|x_n)-\E{p(\mu|X)}) ] \\
\sigma_{n+1} &= \sigma_n + \alpha(n) 
	[ (\E{\mathds{1}_{\sbra{\mu}}X} - \sigma_n) + \\
	&\quad\quad\quad\quad\quad\quad
	(x_n p(\mu|x_n) - \E{x_n p(\mu|X)})  ]
\end{aligned}
\label{eq:rhosigma_sa}
\end{equation}
It is obvious that both stochastic approximation algorithms
satisfy the conditions of Theorem \ref{thm:borkar}.
As a result, they converge to the asymptotic solution of the 
differential equations
\begin{equation*}
\begin{aligned}
\dot \rho &= p(\mu) - \rho \\
\dot \sigma &= \E{\mathds{1}_{\sbra{\mu}}X} - \sigma
\end{aligned}
\end{equation*}
which can be trivially derived through standard ODE analysis to 
be $\pbra{p(\mu), \E{\mathds{1}_{\sbra{\mu}}X}}$.
This follows from the fact that the only 
internally chain transitive invariant sets for (\ref{eq:rhosigma_sa})
are the isolated equilibrium points 
$\pbra{p(\mu), \E{\mathds{1}_{\sbra{\mu}}X}}$.
In other words, we have shown that
\begin{equation}
\pbra{\rho_n,\sigma_n} \xrightarrow{a.s.} \pbra{p(\mu), \E{\mathds{1}_{\sbra{\mu}}X}}
\end{equation}
The convergence of $m_n$ follows from the fact that 
$\E{X|\mu} = \nicefrac{\E{\mathds{1}_{\sbra{\mu}}X}}{p(\mu)}$,
and standard results on the convergence 
of the product of two random variables.
\end{proof}

As a direct consequence of this theorem, the following corollary 
provides an online learning rule that solves the
optimization problem of the deterministic annealing algorithm.

\begin{corollary}
The online training rule 
\begin{equation}
\begin{cases}
\rho_i(n+1) &= \rho_i(n) + \alpha(n) \sbra{ \hat p(\mu_i|x_n) - \rho_i(n)} \\
\sigma_i(n+1) &= \sigma_i(n) + \alpha(n) \sbra{ x_n \hat p(\mu_i|x_n) - \sigma_i(n)}
\end{cases}
\label{eq:oda_learning1}
\end{equation}
where the quantities $\hat p(\mu_i|x_n)$ and $\mu_i(n)$ 
are recursively updated 
as follows:
\begin{equation}
\begin{aligned}
\hat p(\mu_i|x_n) &= \frac{\rho_i(n) e^{-\frac{d(x_n,\mu_i(n))}{T}}}
			{\sum_i \rho_i(n) e^{-\frac{d(x_n,\mu_i(n))}{T}}} \\
\mu_i(n) &= \frac{\sigma_i(n)}{\rho_i(n)},
\end{aligned}
\label{eq:oda_learning2}
\end{equation}
converges almost surely to a possibly sample path dependent solution 
of the optimization (\ref{eq:minFstar}), as $n\rightarrow\infty$.
\end{corollary}

Finally, the determinsitic annealing algorithm with the learning rule
(\ref{eq:oda_learning1}), (\ref{eq:oda_learning2}) can be used to 
define a consistent (histogram)
density estimator at the limit $T\rightarrow 0$.
In the limit $\lambda\rightarrow 0$, and as 
the number of observed samples $\cbra{x_n}$ goes to infinity,
i.e., $n\rightarrow\infty$,
the learning algorithm based on
(\ref{eq:oda_learning1}), (\ref{eq:oda_learning2}),
results in a codevector $\mu$ that 
constructs a consistent density estimator.
This follows from the fact that as $T\rightarrow 0$, 
we get $p^*(\mu_i|x) \rightarrow \mathds{1}_{\sbra{x\in S_i}}$ and
$K\rightarrow\infty$, i.e., the number of effective codevectors $K$ goes to infinity. 
As a result, it can be shown that $Vol(S_i)\rightarrow 0$, where
$S_i = \cbra{x \in S: i = \argmin\limits_j ~ d(x,\mu_j)}$.
Then, 
$\hat p(x) = 
\frac{\sum_i \mathds{1}_{\sbra{x\in S_i}}}{n Vol(S_i)}$
is a consistent density estimator.
The proof follows similar arguments 
to the stochastic vector quantization algorithm 
(see, e.g., \cite{devroye2013})
but is omitted due to space limitations.
%

\subsection{Online Deterministic Annealing for Supervised Learning}

The same learning algorithm 
can be extended for classification as well. 
A multi-class classification problem involves 
a pair of random variables 
$\cbra{X,c} \in S\times S_c$ defined in a probability space
$\pbra{\Omega, \mathcal{F}, \mathbb{P}}$, with
$c\in S_c$ representing the class of $X$ and $S\subseteq\mathbb{R}^d$.	
The codebook is represented by
$\mu := \cbra{\mu_i}_{i = 1}^K$, $\mu_i \in ri(S)$, and
$c_\mu := \cbra{c_{\mu_i}}_{i = 1}^K$,
such that $c_{\mu_i} \in S_c$ 
represents the class of $\mu_i$ for all $i \in \cbra{1,\ldots,K}$.

We can approximate the optimal solution of a minimum classification error
problem by using the distortion measure
\begin{equation}
d^c(x,c_x,\mu,c_\mu) = \begin{cases}
				d(x,\mu),~ c_x=c_\mu \\
				0,~ c_x\neq c_\mu			
				\end{cases}
\label{eq:dc}
\end{equation}
It is easy to see that this particular choice for the distortion measure $d^c$ 
in (\ref{eq:dc}) transforms the learning rule in (\ref{eq:oda_learning1}) to 
\begin{equation}
\begin{cases}
\rho_i(n+1) &= \rho_i(n) + \beta(n) \sbra{ s_i \hat p(\mu_i|x_n) - \rho_i(n)} \\
\sigma_i(n+1) &= \sigma_i(n) + \beta(n) \sbra{ s_i x_n \hat p(\mu_i|x_n) - \sigma_i(n)}
\end{cases}
\label{eq:oda_learning1c}
\end{equation}
where $s_i:=\mathds{1}_{\sbra{c_{\mu_i}=c}}$.
As a result, this is equivalent to estimating strongly consistent
class-conditional density estimators:
\begin{equation}
     \hat p(x|c=j) \rightarrow \pi_j p(x|c=j) ,\ a.s.
\end{equation}
where $\pi_i := \mathbb{P}\sbra{c = i}$.
This results in a Bayes-optimal classification scheme.
The proof is beyond the scope of this paper, since we will focus our attention 
to reinforcement learning.
As a side note, a practical classification rule such as
the nearest-neighbor rule:
\begin{equation}
\hat c(x) = c_{\mu_{h^*}}
\end{equation}
where $h^* = \argmax\limits_{\tau = 
		1,\ldots,K} ~ p(\mu_\tau|x),~ h \in \cbra{1, \ldots, K}$,
results in an easy=to-implement classifier with tight upper bound, i.e.,  
$ J_B^* \leq \hat J_B^* \leq 2 J_B^*$, where 
$J_B$ represents the optimal Bayes error (see, e.g., \cite{devroye2013}).

\subsection{The Algorithm}
\label{sSec:Algorithm}

The Online Deterministic Annealing (ODA) algorithm %
for both clustering and classification 
is shown in Algorithm \ref{alg:ODA}
and its source code is publicly available%
\footnote{
https://github.com/MavridisChristos/OnlineDeterministicAnnealing}.
The temperature parameter $T_i$ is reduced using the geometric series 
$T_{i+1}=\gamma T_i$, for $\gamma<1$.
%
The temperature schedule $\cbra{T_i}$ affects the behavior of the algorithm
by introducing the following trade-off:
small steps $T_i-T_{i-1}$ are theoretically expected to give better results, i.e., 
not miss any bifurcation points, but larger steps provide computational benefits.

\begin{remark}
Notice that the temperature schedule  and its values depend on
the range of the domain of the data.
When the input domain is not known a prior, we can use the formulation 
(\ref{eq:Flambda}), and (\ref{eq:Tlambda}), substituting $T_i$ by 
$
 T_i:=\frac{1-\lambda_i}{\lambda_i}  ,\quad \lambda_i\in\sbra{0,1}.
$
\end{remark}

Regarding the stochastic approximation stepsizes, 
simple time-based learning rates of 
the form $\alpha_n = \nicefrac{1}{a+ bn}$, $a,b>0$,
have experimentally shown to be sufficient 
for fast convergence. 
Convergence is checked with the condition
$T d_\phi(\mu_i^n,\mu_i^{n-1})<\epsilon_c$
for a given threshold $\epsilon_c$.
This condition becomes harder as the value of $T$ decreases.
Exploring adaptive learning rates is among the authors' future 
research direction.
The stopping criteria $SC_{stop}$ can include
a maximum number of codevectors $K_{max}$ allowed,
a minimum temperature $T_{min}$ to be reached, 
a minimum distortion/classification error $e_{target}$ to be reached,
a maximum number of iterations $i_{max}$ reached, and so on.

Bifurcation, at $T_i$, is detected by maintaining a pair of perturbed 
codevectors $\cbra{\mu_j+\delta, \mu_j-\delta}$ 
for each effective codevector $\mu_j$ generated by the algorithm 
at $T_{i-1}$,
i.e. for $j=1\ldots,K_{i-1}$.
Using arguments from variational calculus (see Section \ref{sSec:bifurcation}),
it is easy to see that, upon convegence, 
the perturbed codevectors will merge if a critical 
temperature has not been reached, and will get separated otherwise. 
Therefore, the cardinality of the model is at most doubled at 
every temperature level.
These are the effective codevectors discussed in Section \ref{sSec:bifurcation}.
For classification, a perturbed codevector for each class is generated.
Merging is detected by the condition 
$T d_\phi(\mu_j,\mu_i)<\epsilon_n$,
where $\epsilon_n$ is a design parameter
that acts as a regularization term for the model 
that controls the number of effective codevectors.
These comparisons need not be in any specific order and the worst-case number of 
comparisons is $N_k = \sum_{i=1}^{K-1} i$, which scales with $O(K^2)$.
An additional regularization mechanism
is the detection of idle codevectors, which 
is checked by the condition $\rho_i(n)<\epsilon_r$, 
where $\rho_i(n)$ can be seen as an
approximation of the probability $p(\mu_i,c_{\mu_i})$.
In practice, $\epsilon_c$, $\epsilon_n$, $\epsilon_r$ are assigned similar values and 
their impact on the performance is similar to any threshold parameter that detects convergence.

The complexity of Alg. \ref{alg:ODA} for a fixed temperature coefficient $T_i$ is 
$O(N_{c_i} (2K_i)^2 d)$,
where $N_{c_i}$ is the number of stochastic approximation iterations needed for convergence 
which corresponds to the number of data samples observed, 
$K_i$ is the number of codevectors of the model at temperature $T_i$, and 
$d$ is the dimension of the input vectors, i.e., $x\in\mathbb{R}^d$.
Therefore, assuming a training dataset of $N$ samples and
a temperature schedule $\cbra{ T_1=T_{max}, T_2, \ldots, T_{N_T}=T_{min} }$, 
the worst case complexity of Algorithm \ref{alg:ODA} becomes:
\begin{align*}
O(N_{c} (2\bar K)^2 d)
\end{align*}
where $N_c=\max_i \cbra{N_{c_i}}$ is an upper bound on the number of data samples observed
until convergence at each temperature level, and
\begin{align*}
N_T \leq \bar K \leq \min\cbra{ \sum_{n=0}^{N_T-1} 2^n, \sum_{n=0}^{\log_2 K_{max}} 2^n}
< N_T K_{max}
\end{align*}
where the actual value of $\bar K$
depends on the bifurcations occurred as a result of reaching critical temperatures
and the effect of the regularization mechanisms described above.
Note that typically $N_c \ll N$ as a result of the stochastic approximation algorithm,
and $\bar K \ll N_T K_{max}$ as a result of the progressive nature of the ODA 
algorithm.

As a final note, because the convergence to the Bayes decision surface 
comes in the limit $(K,T)\rightarrow (\infty,0)$, 
in practice, a fine-tuning mechanism can be designed to run on top of 
Alg. \ref{alg:ODA} after a predefined threshold temperature $T_{min}$.
This can be either an LVQ algorithm \cite{sato1996generalized}
or some other local model, 
i.e., we can use the partition created by Alg. \ref{alg:ODA}
to train local models in each region of the data space.


\begin{algorithm}[hb!]
\caption{Online Deterministic Annealing (ODA)} 
\label{alg:ODA}
\begin{algorithmic}
\STATE Select a Bregman divergence $d_\phi$
\STATE Set stopping criteria $SC_{stop}$ (e.g., $K_{max}$, $T_{min}$)
\STATE Set convergence parameters: $\gamma$, $\cbra{\alpha_n}$, 
	$\epsilon_c$, $\epsilon_n$, $\epsilon_r$, $\delta$ 
\STATE Initialize:	$K = 1$, $\lambda = 1-\epsilon$, $T = \nicefrac{1-\lambda}{\lambda}$,
$\cbra{\mu^i}$, $\cbra{c_{\mu^i} = c,~ \forall c \in\mathcal{C} }$, 
$\cbra{p(\mu^i) = 1}$, $\cbra{\sigma(\mu^i) = \mu^i p(\mu_i)}$ 
\REPEAT
\STATE Perturb  
	$\mu^i \gets  
		\cbra{\mu^i+\delta, \mu^i-\delta}$, $\forall i$ 
\STATE Update $K\gets 2K$,
$p(\mu^i)$, $\sigma(\mu^i)\gets\mu^i p(\mu^i)$, $\forall i$ 
\STATE $n \gets 0$
\REPEAT 
%
\STATE Observe data point $x$ and class label $c$
\FOR{$i = 1,\ldots, K$} 
\STATE Compute membership $s^i = \mathds{1}_{\sbra{c_{\mu^i}=c}}$ 
\STATE Update: 
\vskip -0.3in
	\begin{align*}
	p(\mu^i|x) &\gets \frac{p(\mu^i) e^{-\frac{d_\phi(x,\mu^i)}{T}}}
			{\sum_j p(\mu^j) e^{-\frac{d_\phi(x,\mu^j)}{T}}} \\
	p(\mu^i) &\gets p(\mu^i) + \alpha_n \sbra{s^i p(\mu^i|x) - p(\mu^i)} \\
	\sigma(\mu^i) &\gets \sigma(\mu^i) + 
		\alpha_n \sbra{s^i x p(\mu^i|x) - \sigma(\mu^i)} \\
	\mu^i &\gets \frac{\sigma(\mu^i)}{p(\mu^i)}	
	\end{align*}
\vspace{-1.5em}
\STATE $n\gets n+1$
\ENDFOR
\UNTIL Convergence: $T d_\phi(\mu^i_n,\mu^i_{n-1})<\epsilon_c$, $\forall i $
\STATE Keep effective codevectors: 
    discard $\mu^i$ if $T d_\phi(\mu^j,\mu^i)<\epsilon_n$, 
	$\forall i,j,i\neq j$
\STATE Remove idle codevectors: 
    discard $\mu^i$ if $p(\mu^i)<\epsilon_r$, $\forall i$
\STATE Update $K$, $p(\mu^i)$, $\sigma(\mu^i)$, $\forall i$
\STATE Lower temperature: $\lambda \gets \gamma \lambda$, $T \gets \frac{1-\lambda}{\lambda}$  
\UNTIL $SC_{stop}$
\end{algorithmic}
\end{algorithm}

\section{Online Deterministic Annealing for Reinforcement Learning}
\label{Sec:RL}

The learning architecture of Alg. \ref{alg:ODA}
can also be integrated with reinforcement learning methods, giving rise 
to adaptive state-action aggregation schemes.
As will be shown in this section, this is a result of 
using stochastic approximation as a training rule, and 
yields a 
reinforcement learning algorithm based on a progressively changing 
underlying model \cite{mavridis2021vector,mavridis2021maximum}.
%


We consider an MDP $(\mathcal{X}, \mathcal{U}, \mathcal{P}, C)$,
where $S\subseteq \mathbb{R}^{d_X}$,
$S\subseteq \mathbb{R}^{d_U}$ are compact convex sets
(see Section \ref{sSec:Qlearning}).
We are interested in the approximation of the quality function 
$Q:\mathcal{X}\times\mathcal{U}\rightarrow\mathbb{R}_+$.
%
%
To this end, we use the online deterministic annealing (ODA)
algorithm (Alg. \ref{alg:ODA}) as an online recursive algorithm that finds
an optimal representation of the data space with respect to a
trade-off between minimum average distortion and maximum entropy.  
We define a quantizer 
$Q_P(x,u) = \sum_{h=1}^K \mu_h \mathds{1}_{\sbra{(x,u) \in P_h}}$, 
where $\cbra{P_h}_{h=1}^K$ is a partition of $\mathcal{X}\times \mathcal{U}$.
The parameters $\mu_h:=(m_h,v_h)$ define a state-action aggregation scheme 
with $K$ clusters (aggregate state-action pairs), each represented by 
$m_h\in \mathcal{X}$ and $v_h\in \mathcal{U}$, for $h=1,\ldots, K$.
After convergence, if the representation is meaningful, 
the finite set $\cbra{\mu_h}_{h=1}^K$,
where $\mu_h\in \mathcal{X}\times \mathcal{U}$, can be used 
directly as a piece-wise constant approximation of the $Q$ function.
We stress that the cardinality $K$ of the set of representatives of the space
$\mathcal{X}\times\mathcal{U}$ is automatically updated by Alg. \ref{alg:ODA}
and progressively increases, as needed, with respect to the
complexity-accuracy trade-off presented above.

\begin{remark}
It is also possible to use $\cbra{\mu_h}_{h=1}^K$ as pseudo-inputs
for an adaptive and sparse Gaussian process regression \cite{mavridis2022sparse}, 
but this is beyond the scope of this paper.
\end{remark}


In essence, we are approximating the $Q$ function
with a piece-wise constant parametric model 
with the parameters that define the partition living in the data space
and being chosen by the online deterministic annealing algorithm (Alg. \ref{alg:ODA}).
However, since the system observes its states and actions 
online while learning its optimal policy 
using a temporal-difference 
reinforcement learning algorithm,
the two estimation algorithms need to run at the same time.
This can become possible by observing that 
Algorithm \ref{alg:ODA}, as well as most temporal-difference algorithms, 
are stochastic approximation algorithms.
Therefore, we can design a reinforcement learning algorithm as 
a two-timescale stochastic approximation algorithm with 
(a) a fast component that updates the values $Q:=\cbra{Q(h)}_{h=1}^K$
with a temporal-difference learning algorithm, 
and (b) a slow component that updates the representation
$\mu:=\cbra{\mu_h}_{h=1}^K$ based on Alg. \ref{alg:ODA}.
Such a framework can incorporate different reinforcement
learning algorithms, including the proposed algorithm
presented in Alg. \ref{alg:QlearningODA}.
The exploration policy $\pi_L(h|\mu)$ in Alg. \ref{alg:QlearningODA}
depends on the aggregate state
$h$ and balances the ratio between 
exploration and exploitation. 

\begin{algorithm}
\caption{Reinforcement Learning with ODA}
\begin{algorithmic}
\STATE Initialize $\mu_h$, $Q_0(h)$, $\forall h\in\cbra{1,\ldots,K}$
\REPEAT 
	\STATE Observe $x$ 
	and find 
	$$h= \argmin\limits_{\tau = 1,\ldots,k} ~ d_{\phi} ((x,u'),\mu_\tau)$$	
	\STATE Choose $u'=\pi_L(h|\mu)$	 
	\STATE Observe $x'=f(x,u')$ 
	and find
	$$h'= \argmin\limits_{\tau = 1,\ldots,k} ~ d_{\phi} (x',\mu(\tau))$$  
	\STATE Update $Q(h)$: 
	\begin{align*}
	 Q_{i+1}(h) = Q_{i}(h) &+ \alpha_i [ C(x,u')\\
		  &~~~~~~~ + \gamma \min_{u} Q_{i}(h') - Q_{i}(h) ]
	\end{align*} 
	\IF{$i \mod N = 0$}	
	\STATE Update partition $\mu:=\cbra{\mu_h}_{h=1}^K$ 
		using Alg. \ref{alg:ODA}
	\ENDIF
\UNTIL Convergence 
\STATE Update Policy:
 $$ u^*(x) = \argmin_u \cbra{~ Q(h(x)) ~} $$
\end{algorithmic}
\label{alg:QlearningODA}
\end{algorithm}
%


The convergence properties of the algorithm 
can be studied by directly applying the theory of the O.D.E. method 
for stochastic approximation in multiple timescales 
in Theorem \ref{thm:borkar_timescales}.
For more details see
\cite{borkar2009stochastic}. 
As a result, Alg. \ref{alg:QlearningODA}
converges according to the following theorem:

\begin{theorem}
Algorithm \ref{alg:QlearningODA} converges almost surely to 
$(\mu^*,Q^*)$ where $\mu*$ is a solution of the 
optimization problem (\ref{eq:minFstar}),
and $Q^*$ minimizes the temporal-difference error:
%
{\small
\begin{equation}
J_h = \| \E{C(x,u)+\gamma \min_u Q(h')|(x,u)\in P_h} - Q(h) \|^2
\end{equation}
}
where $h=1,\ldots,K$, 
and $\cbra{P_h}_{h=1}^K$ is a partition of $\mathcal{X}\times\mathcal{U}$
with every $P_h$ assumed to be visited infinitely often.
\label{prop:Convergence}
\end{theorem}

\begin{proof}
From Theorem \ref{thm:oda_sa}, it follows that Algorithm (\ref{alg:ODA}) is a  
stochastic approximation algorithm of the form
(\ref{eq:sa_timescales_y}) that converges 
to a solution of (\ref{eq:minFstar}). 
It is easy to see that the $Q$ function update in 
Alg. \ref{alg:QlearningODA} 
is also a stochastic approximation
algorithm of the form (\ref{eq:sa_timescales_x}), 
for $f(Q(h)) = -\nabla_{Q(h)} J_h$.
The result follows from Theorem \ref{thm:borkar_timescales}.
\end{proof}

We note that the condition 
$\nicefrac{\beta_i}{\alpha_i}\rightarrow 0$ is of great importance.
Intuitively, Algorithm \ref{alg:QlearningODA} 
consists of two components running in different timescales.
The slow component updates $\mu$ 
and is viewed as quasi-static 
when analyzing the behavior of the fast transient $Q$ which 
updates the approximation of the quality function.
As an example, the condition $\nicefrac{\beta_n}{\alpha_n}\rightarrow 0$
is satisfied by stepsizes of the form 
$(\alpha_n,\beta_n)=(\nicefrac 1 n, \nicefrac{1}{1+n \log n})$, or
$(\alpha_n,\beta_n)=(\nicefrac{1}{n^{\nicefrac 2 3}}, \nicefrac{1}{n})$.
Another way of achieving the two-timescale effect is to 
run the iterations for the slow component $\cbra{\mu_n}$ 
with stepsizes $\cbra{\alpha_{n(k)}}$, 
where $n(k)$ is a subsequence of $n$ that becomes increasingly rare
(i.e. $n(k+1)-n(k)\rightarrow\infty$), 
while keeping its values constant between these instants.
In practice, it has been observed that a good policy is to 
run the slow component with slower stepsize schedule $\beta_n$ and 
update it along a subsequence keeping its values constant in between
(\cite{borkar2009stochastic}, Ch. 6).
This explains the parameter $N$ in Alg. \ref{alg:QlearningODA}
whose value should increase with time.

\begin{remark}
Alg. \ref{alg:QlearningODA} is essentially based on successive entropy-regularized 
reinforcement learning problems. 
However, the entropy is defined with respect to the learning representation in the 
state-action space $\mathcal{X}\times \mathcal{U}$. 
As such, this approach is not to be directly compared to common entropy-regularized
approaches as in \cite{haarnoja2017}, and related methods, such as PPO \cite{schulman2017proximal}.
The deeper connection to risk-sensitive 
reinforcement learning can be studied 
along the lines of \cite{mavridis2022risk} and \cite{noorani2023risk}.
\end{remark}

\section{Experimental Evaluation and Discussion}
\label{Sec:Results}

We illustrate the properties and evaluate the performance 
of the proposed algorithm in 
supervised, unsupervised, and reinforcement learning problems.

\subsection{Supervised and Unsupervised Learning}

\begin{figure*}[h]
\centering
\includegraphics[trim=0 0 0 0,clip,width=0.19\textwidth]{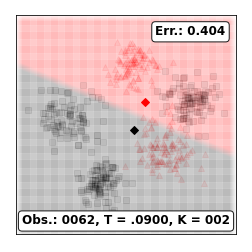}
\includegraphics[trim=0 0 0 0,clip,width=0.19\textwidth]{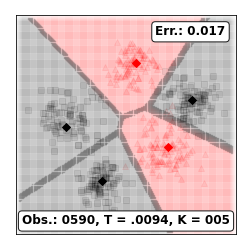}
\includegraphics[trim=0 0 0 0,clip,width=0.19\textwidth]{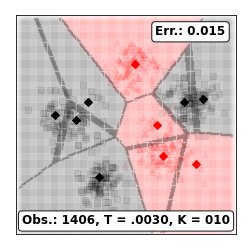}
\includegraphics[trim=0 0 0 0,clip,width=0.19\textwidth]{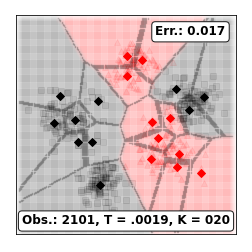}
\includegraphics[trim=0 0 0 0,clip,width=0.19\textwidth]{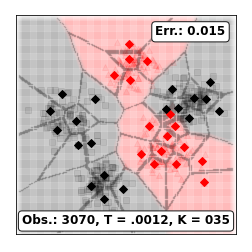}
\caption{Evolution of the ODA algorithm in 2D binary classification based on class-conditional density estimation.
Temperature $T$ decreases from left to right.}
\label{fig:domain-1}
\end{figure*}

\begin{figure}[h]
\centering
\includegraphics[trim=0 0 0 0,clip,width=0.24\textwidth]{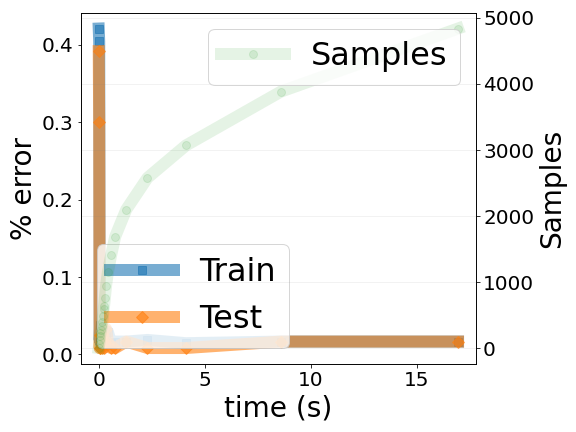}
\includegraphics[trim=0 0 0 0,clip,width=0.24\textwidth]{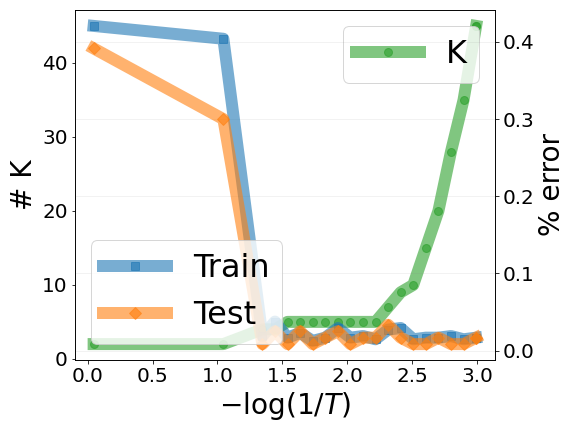}
\caption{Performance curves for the problem of Fig. \ref{fig:domain-1}.}
\label{fig:domain-1pc}
\end{figure}

\begin{figure}[h]
\centering
\centering
\includegraphics[trim=280 200 250 220,clip,width=0.16\textwidth]{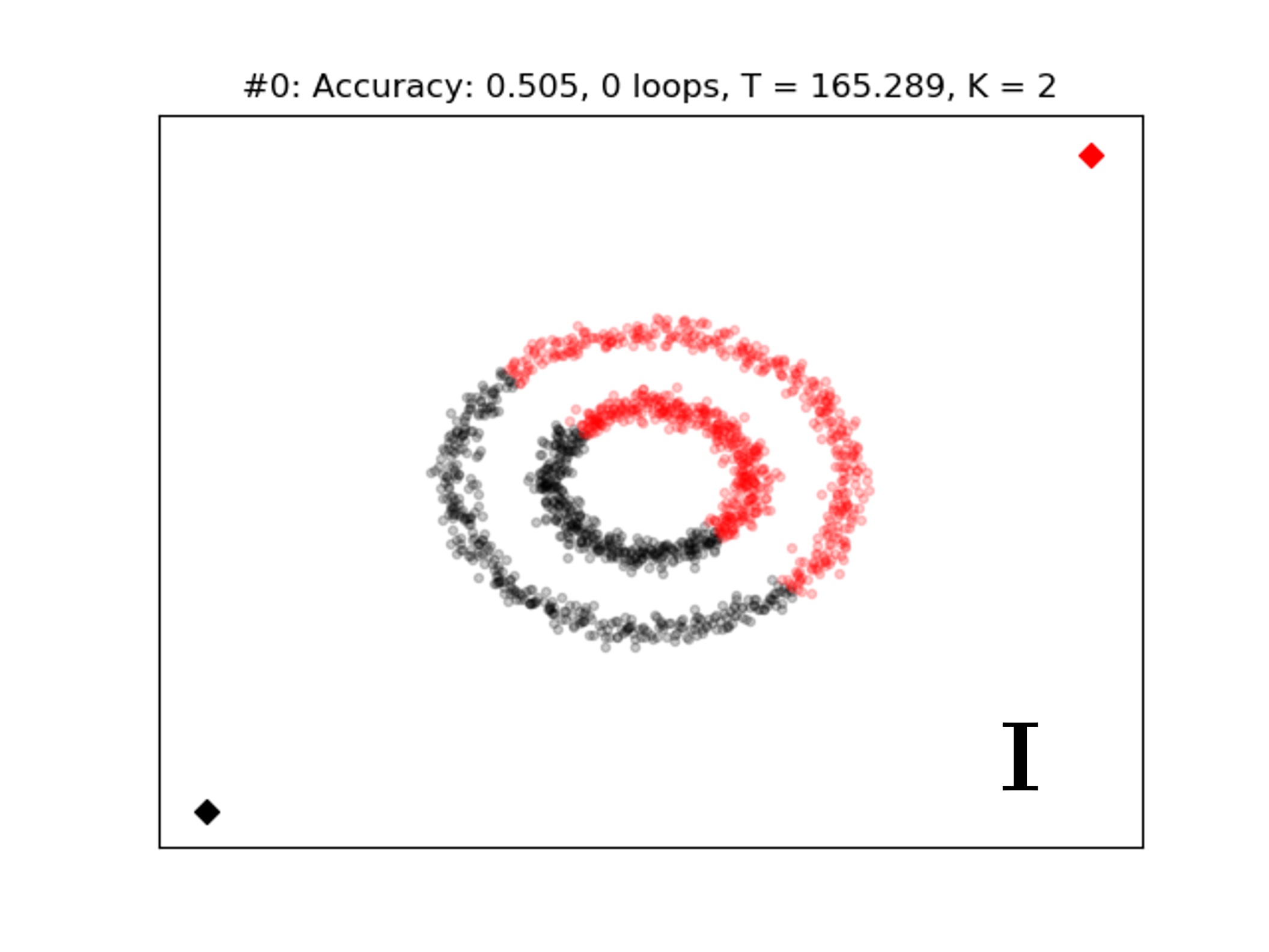}
\hspace{1em}
\includegraphics[trim=280 200 250 220,clip,width=0.16\textwidth]{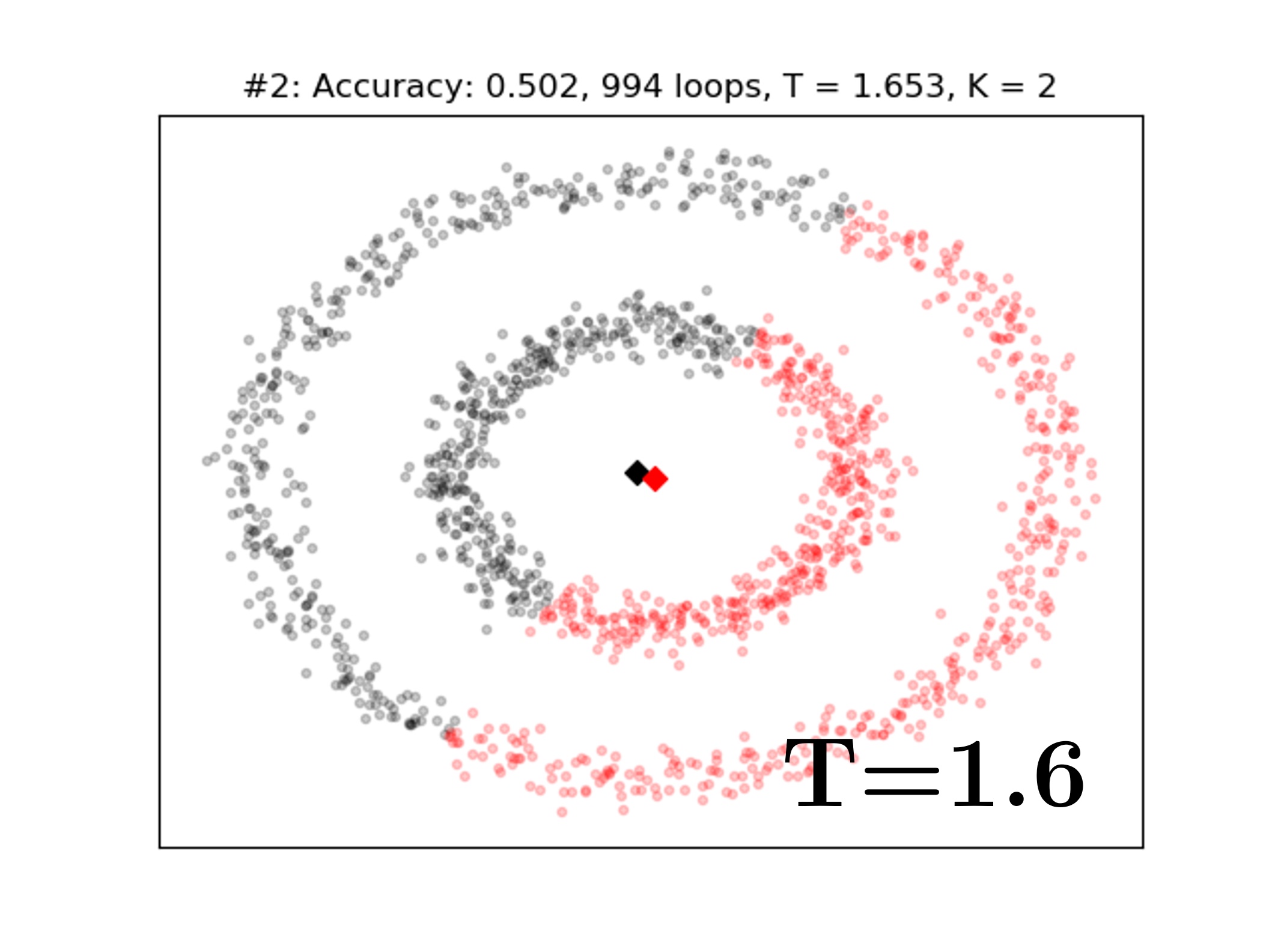}
\includegraphics[trim=280 200 250 220,clip,width=0.16\textwidth]{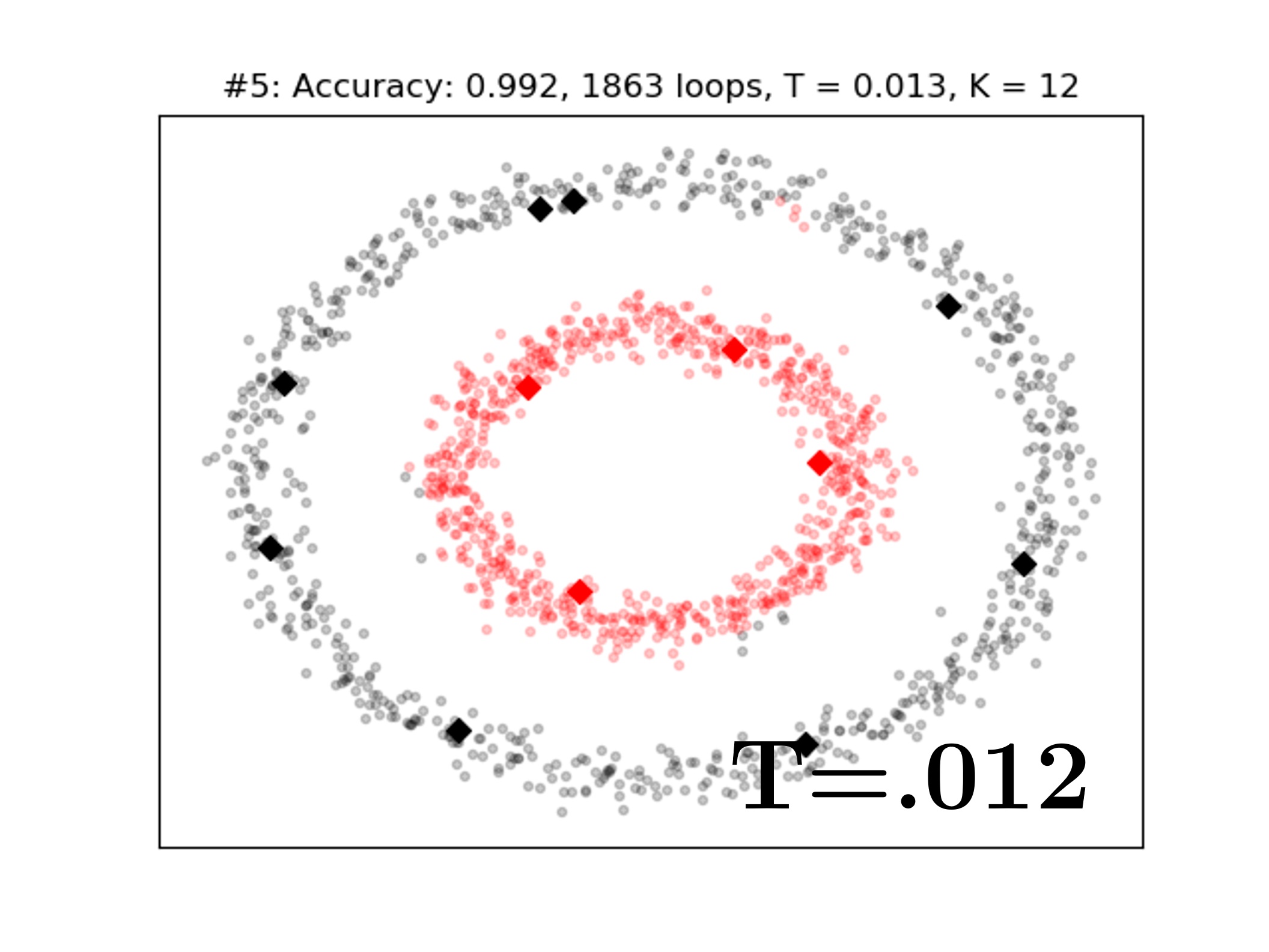}
\hspace{1em}
\includegraphics[trim=280 200 250 220,clip,width=0.16\textwidth]{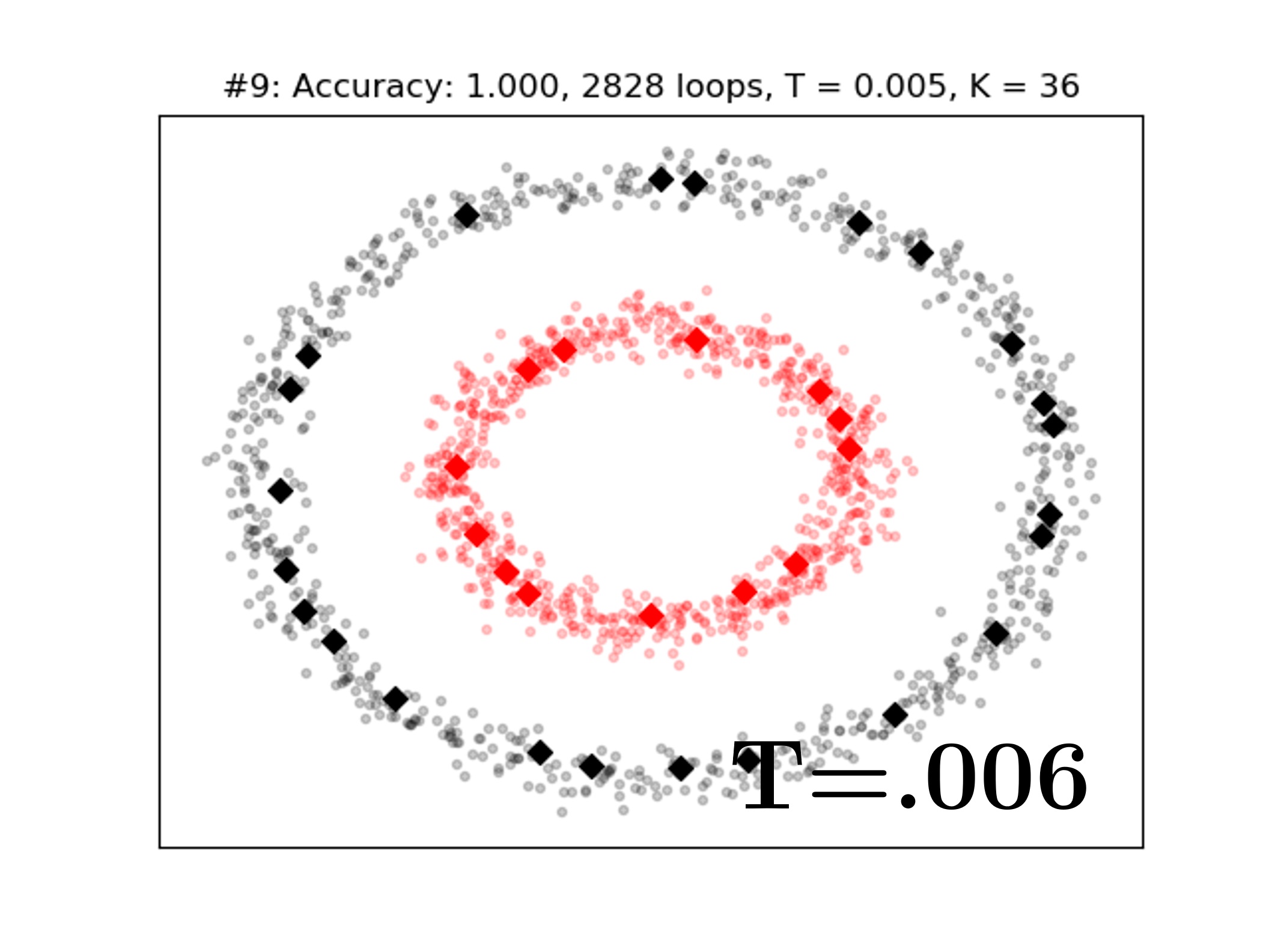}
\caption{Illustration of the evolution of Alg. \ref{alg:ODA} 
for binary classification in 2D based on class-conditional density estimation.
Showcasing robustness to bad initial conditions which idicates poor local minima prevention.}
\label{fig:illustration}
\end{figure}

We first illustrate the properties of Alg. \ref{alg:ODA}, 
in a classification problem where
the data samples were sampled from a mixture of 2D 
Gaussian distributions.
In Fig. \ref{fig:domain-1} and \ref{fig:domain-1pc}, the temperature level (the values of $T$ shown are the normalized values $\lambda=\nicefrac{1}{T+1}$),
the average distortion of the model, 
the number of codevectors (neurons) used, 
the number of observations (data samples) used for convergence, as well as the overall time, are shown.
Since the objective is to give a geometric illustration 
of how the algorithm works in the two-dimensional plane, 
the Euclidean distance is used as the proximity measure.
Notice that the classification accuracy for $K=5$ is $98\%$ and it gets to 
$100\%$ only when we reach $K=144$. 
This showcases the performance-complexity trade-off that Alg. \ref{alg:ODA}
allows the user to control in an online fashion.
Since classification and clustering are handled in a similar way by Alg. \ref{alg:ODA},
these examples properly illustrate the behavior of the proposed methodology for clustering as well.
In Fig. \ref{fig:illustration}, the progression of the learning representation is depicted 
for a binary classification problem with
underlying class distributions shaped as
concentric circles.
The algorithm starts at high temperature 
with a single codevector for each class. 
Here the codevectors are poorly initialized 
outside the support of the data, which is not assumed known a priori
(e.g. online observations of unknown domain).
In this example the LVQ algorithm has been shown to fail \cite{aLaVigna_LVQconvergence_1990}.
This showcases the
robustness of the proposed algorithm with respect to 
the initial configuration.
This is an example of poor local minima prevention which, although not theoretically guaranteed,  is a known property of annealing optimization methods.
%


%
\begin{figure}[h]
\centering
\begin{subfigure}[b]{0.24\textwidth}
\centering
\includegraphics[trim=0 0 0 0,clip,width=\textwidth]{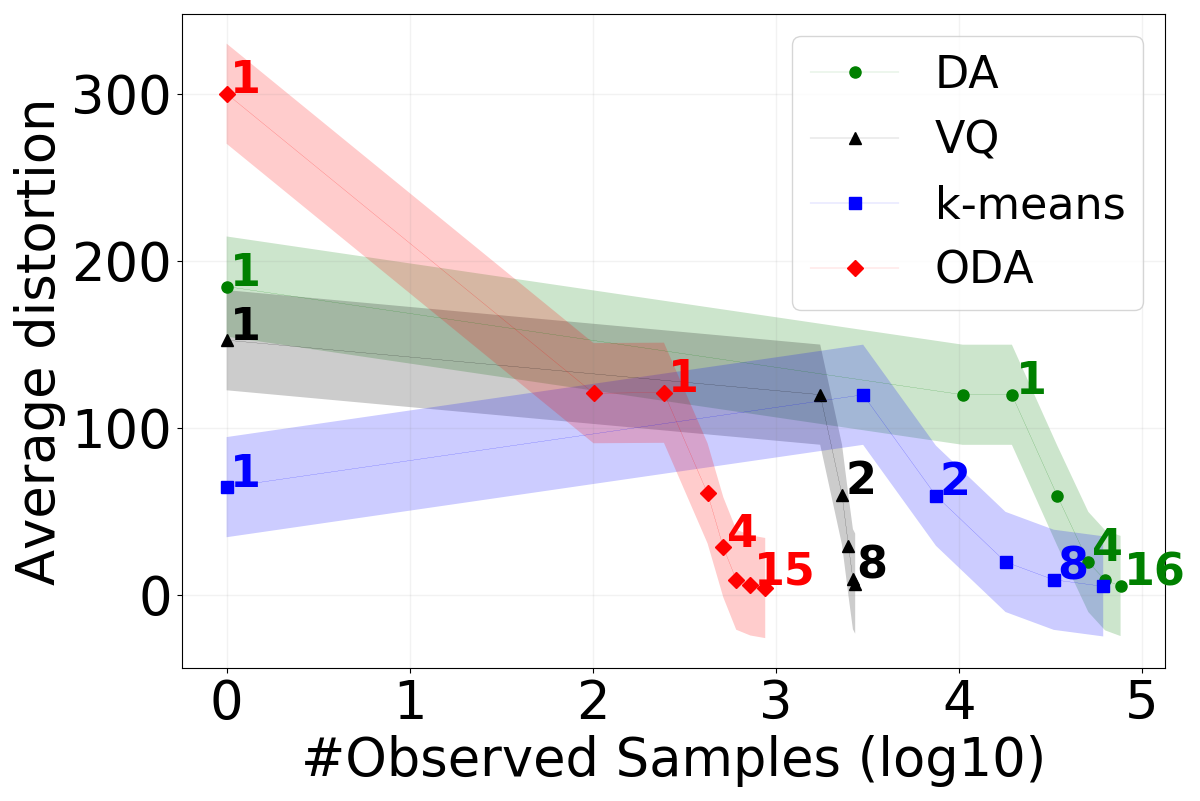}
\caption{Gaussians.}
\label{sfig:clustering_blobs}
\end{subfigure}
%
%
\hfill
\begin{subfigure}[b]{0.24\textwidth}
\centering
\includegraphics[trim=0 0 0 0,clip,width=\textwidth]{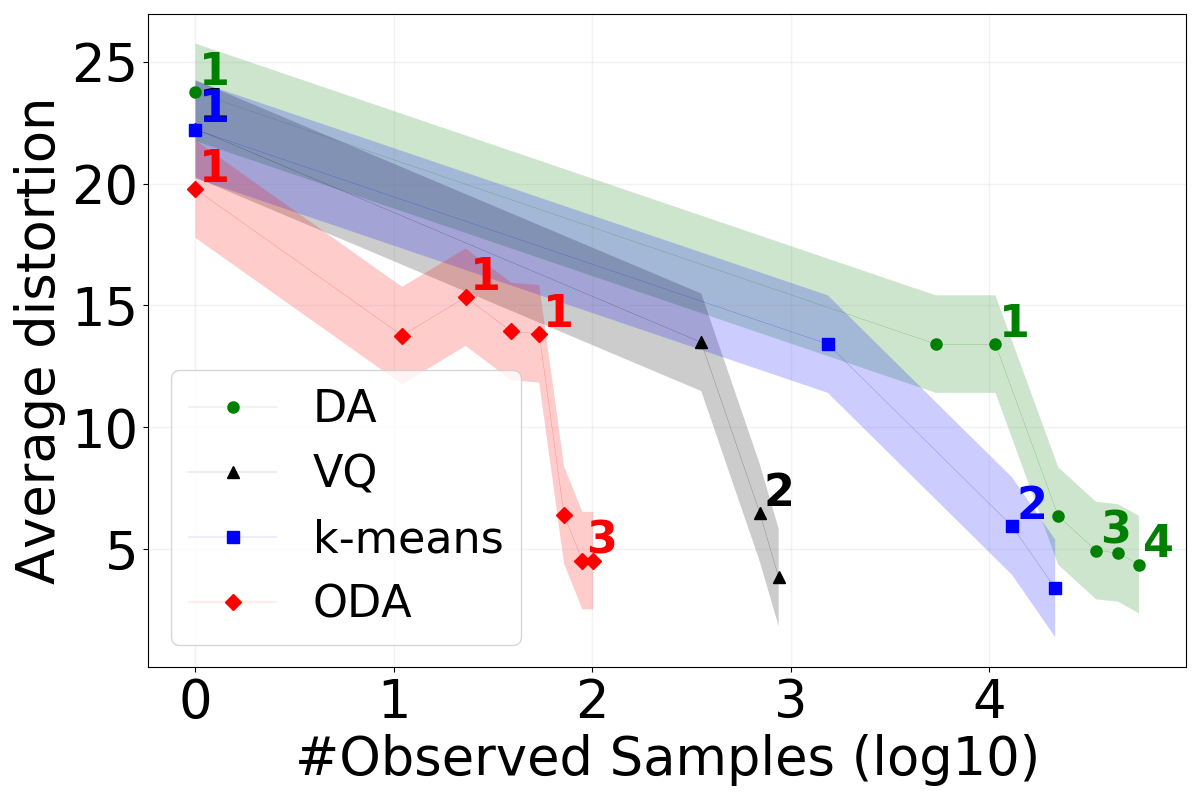}
\caption{PIMA.}
\label{sfig:clustering_pima}
\end{subfigure}
%
%
\caption{Algorithm comparison for unsupervised learning.}
\label{fig:clustering}
\end{figure}

For clustering, we consider 
the dataset of Fig. \ref{fig:domain-1} (Gaussians), and
the PIMA dataset \cite{smith1988using}.
%
In Fig. \ref{fig:clustering}, 
we compare Alg. \ref{alg:ODA}
with the stochastic (online) vector quantization (sVQ) algorithm (\cite{aLaVigna_LVQconvergence_1990}),
and two offline (batch) algorithms, namely
$k$-means \cite{bottou1995convergence},
and the original 
deterministic annealing (DA) algorithm \cite{rose1998deterministic}.
The algorithms are compared in terms of the minimum average distortion 
achieved, as a function of the number of samples they observed,
and the number of clusters they used.
The Euclidean distance is used for fair comparison.
Since there is no criterion to decide the number of clusters $K$
for $k$-means and sVQ, we run them sequentially for the $K$ values
estimated by DA, and add up the computational time.
All algorithms are able to achieve 
comparable average distortion values given good initial conditions 
and appropriate size $K$. 
Therefore, the progressive estimation of $K$, as well as the robustness 
with respect to the initial 
conditions, are key features of both annealing algorithms.
Compared to the offline algorithms, i.e., $k$-means and DA, 
ODA and sVQ achieve practical convergence
with significantly lower number of 
observations, which corresponds to 
reduced computational time, as argued above. 
Compared to the online sVQ (and LVQ), 
the probabilistic approach of ODA introduces additional computational cost:  
all neurons are now updated in every iteration, 
instead of only the winner neuron.
However, the updates can still be computed relatively fast when using
Bregman divergences (Theorem \ref{thm:bregman_in_DA}). 
%
For more experimental results regarding clustering and classification, 
the authors are referred to \cite{mavridis2022online,mavridis2021progressive,mavridis2022sparse}.

\subsection{Reinforcement Learning}

Finally, we validate the proposed methodology 
on the inverted pendulum (Cart-pole) optimal control problem.
The state variable of the cart-pole system 
has four components $(x,\theta,\dot x,\dot\theta)$, where
$x$ and $\dot x$ are the position and velocity of the cart on the track, and 
$\theta$ and $\dot \theta$ are the angle and angular velocity 
of the pole with the vertical.
%
The cart is free to move within
the bounds of a one-dimensional track. The pole is free to
move only in the vertical plane of the cart and track. 

%
%

The action space consists of an impulsive ``left'' or ``right'' force 
$F\in \{-10,+10\}$N of fixed magnitude to the cart at discrete time intervals.
The cart-pole system is modeled by the following 
nonlinear system of differential equations \cite{barto1983neuronlike}:
{\small
\begin{align*}
\ddot x &= \frac{F + m l \pbra{\dot \theta^2\sin\theta -\ddot \theta \cos \theta}
				\mu_c sgn(\dot x)}
			{m_c + m} \\
\ddot \theta &= \frac{g\sin\theta + \cos\theta 
				\pbra{ \frac{-F - ml\dot\theta^2\sin\theta + \mu_c sgn(\dot x)}
					{m_c+m} }- \frac{\mu_p}{ml}
				    \mu_c sgn(\dot x)}
			{l \pbra{ \frac{4}{3} -\frac{m\cos^2 \theta}{m_c+m} }}			
\end{align*}
}
where the parameter values for $g,m_c,m,l,\mu_c,\mu_p$ can be found in 
\cite{barto1983neuronlike}.
%
%
%
The transition function for the state $x$ is 
$x_{n+1} = x_n + \tau \dot x$, 
%
%
where $\tau=0.02$s. 
The initial state is set to 
$X_0=(u_x,u_\theta,u_{\dot x},u_{\dot\theta})$
where $u_x$, $u_\theta$, $u_{\dot x}$, and 
$u_{\dot\theta}$ follow a uniform distribution $U(-0.05,0.05)$. 
Failure occurs when $|\theta|>12^\circ$ or when $|x|>2.4$m. 
An episode terminates successfully after $N_t$ timesteps, and the 
average number of timesteps $\hat N_t \leq N_t$ 
across different attempts, is used 
to quantify the performance of the learning algorithm.
We use the Euclidean distance as the Bregman divergence $d_\phi$.

In Fig. \ref{fig:oda_results} 
we compare the 
average number of timesteps 
(here $N_t=1000$) with respect to the number of aggregate states used,
for three different state aggregation algorithms.
The first one is naive discretization without state aggregation, 
the second is the SOM-based algorithm proposed in \cite{mavridis2021vector},
and the last is the proposed algorithm Alg. \ref{alg:QlearningODA}.
We initialize the codevectors $\mu$ by uniformly discretizing over 
$\hat S\times \cbra{-10,10}$, for 
$\hat S = [-1,1]\times[-4,4]\times[-1,1]\times[-4,4]$.
We use  
$K\in\cbra{16,81,256,625}$ clusters, 
corresponding to a standard discretization scheme
with only $n\in\cbra{2,3,4,5}$ bins for each dimension.
As expected, state aggregation outperforms standard discretization
of the state-action space.
The ability to progressively adapt the number 
and placement of the centroids of the aggregate states 
is an important property of the proposed algorithm.
Fig. \ref{fig:oda_results} shows $5$ instances
of Alg. \ref{alg:QlearningODA} for $5$ different
stopping criteria according to a pre-defined minimum temperature $T_{min}$.
This results in different representations of the state space with  
$K\in\cbra{56,118,136,202,252}$ aggregate states, respectively.
%

%
\begin{figure}[h]
\centering
\includegraphics[trim=0 0 0 60,clip,width=0.4\textwidth]%
								{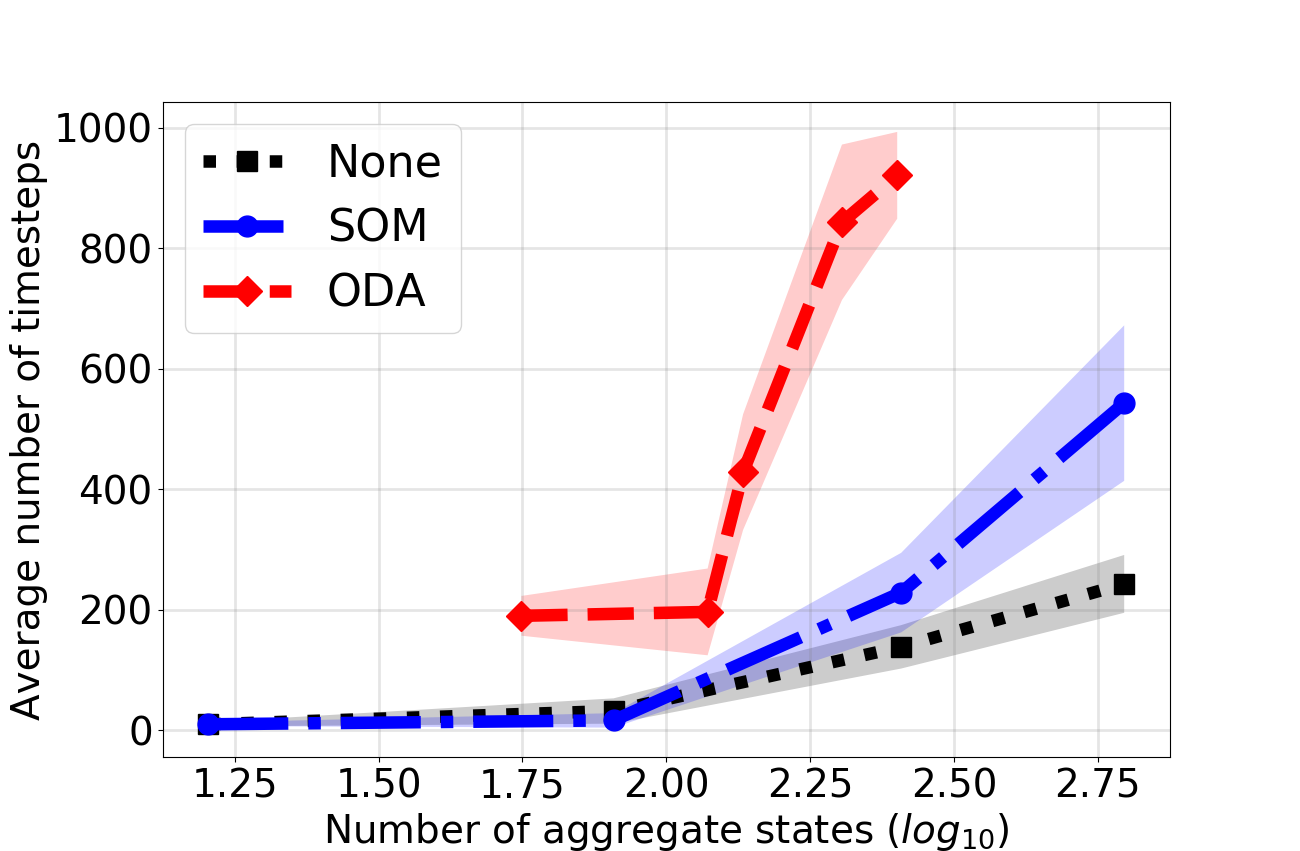}
\caption{Average number of timesteps 
	($N_t=1000$) over number of aggregate states used.
	(red) the proposed algorithm. (black) 
	$Q$-learning without state aggregation. (blue) 
	the SOM-based algorithm of \cite{mavridis2021vector}.}
\label{fig:oda_results}
\end{figure}
%

\section{Conclusion}
\label{Sec:Conclusion}

%
We investigate the properties of learning with progressively growing models,
and propose an online annealing optimization approach
as a learning algorithm
that progressively adjusts its complexity with respect to new observations, 
offering online control over the performance-complexity trade-off.
We show that the proposed algorithm constitutes a 
progressively growing competitive-learning neural network
with inherent regularization mechanisms, 
the learning rule of which is formulated as an online 
gradient-free stochastic approximation algorithm.
%
%
The use of stochastic approximation enables the study 
of the convergence of the learning algorithm through 
mathematical tools from dynamical systems and control, 
and allows for its use in
supervised, unsupervised, and reinforcement learning settings.
In addition, the annealing nature of the algorithm,
contributes to poor local minima prevention and
offers robustness with respect to the initial conditions.
To our knowledge, this is the first time such a progressive approach 
has been proposed for machine learning and reinforcement learning applications.
These results can lead to new developments
in the development of 
progressively growing machine learning models
targeted towards applications in which computational resources are limited and robustness
and interpretability are prioritized.

\bibliographystyle{IEEEtran} %
\bibliography{bib_learning.bib,bib_mavridis.bib}

\begin{thebibliography}{10}
\providecommand{\url}[1]{#1}
\csname url@samestyle\endcsname
\providecommand{\newblock}{\relax}
\providecommand{\bibinfo}[2]{#2}
\providecommand{\BIBentrySTDinterwordspacing}{\spaceskip=0pt\relax}
\providecommand{\BIBentryALTinterwordstretchfactor}{4}
\providecommand{\BIBentryALTinterwordspacing}{\spaceskip=\fontdimen2\font plus
\BIBentryALTinterwordstretchfactor\fontdimen3\font minus
  \fontdimen4\font\relax}
\providecommand{\BIBforeignlanguage}[2]{{%
\expandafter\ifx\csname l@#1\endcsname\relax
\typeout{** WARNING: IEEEtran.bst: No hyphenation pattern has been}%
\typeout{** loaded for the language `#1'. Using the pattern for}%
\typeout{** the default language instead.}%
\else
\language=\csname l@#1\endcsname
\fi
#2}}
\providecommand{\BIBdecl}{\relax}
\BIBdecl

\bibitem{bennett2006interplay}
K.~P. Bennett and E.~Parrado-Hern{\'a}ndez, ``The interplay of optimization and
  machine learning research,'' \emph{The Journal of Machine Learning Research},
  vol.~7, pp. 1265--1281, 2006.

\bibitem{lecun2015deep}
Y.~LeCun, Y.~Bengio, and G.~Hinton, ``Deep learning,'' \emph{nature}, vol. 521,
  no. 7553, pp. 436--444, 2015.

\bibitem{thompson2020computational}
N.~C. Thompson, K.~Greenewald, K.~Lee, and G.~F. Manso, ``The computational
  limits of deep learning,'' \emph{arXiv preprint arXiv:2007.05558}, 2020.

\bibitem{strubell2019energy}
E.~Strubell, A.~Ganesh, and A.~McCallum, ``Energy and policy considerations for
  deep learning in nlp,'' \emph{arXiv preprint arXiv:1906.02243}, 2019.

\bibitem{szegedy2013intriguing}
C.~Szegedy, W.~Zaremba, I.~Sutskever, J.~Bruna, D.~Erhan, I.~Goodfellow, and
  R.~Fergus, ``Intriguing properties of neural networks,'' \emph{arXiv preprint
  arXiv:1312.6199}, 2013.

\bibitem{carlini2017towards}
N.~Carlini and D.~Wagner, ``Towards evaluating the robustness of neural
  networks,'' in \emph{2017 ieee symposium on security and privacy (sp)}.\hskip
  1em plus 0.5em minus 0.4em\relax IEEE, 2017, pp. 39--57.

\bibitem{sehwag2019analyzing}
V.~Sehwag, A.~N. Bhagoji, L.~Song, C.~Sitawarin, D.~Cullina, M.~Chiang, and
  P.~Mittal, ``Analyzing the robustness of open-world machine learning,'' in
  \emph{Proceedings of the 12th ACM Workshop on Artificial Intelligence and
  Security}, 2019, pp. 105--116.

\bibitem{xu2012robustness}
H.~Xu and S.~Mannor, ``Robustness and generalization,'' \emph{Machine
  learning}, vol.~86, no.~3, pp. 391--423, 2012.

\bibitem{northcutt2021pervasive}
C.~G. Northcutt, A.~Athalye, and J.~Mueller, ``Pervasive label errors in test
  sets destabilize machine learning benchmarks,'' \emph{arXiv preprint
  arXiv:2103.14749}, 2021.

\bibitem{raghunathan2019adversarial}
A.~Raghunathan, S.~M. Xie, F.~Yang, J.~C. Duchi, and P.~Liang, ``Adversarial
  training can hurt generalization,'' \emph{arXiv preprint arXiv:1906.06032},
  2019.

\bibitem{mavridis2020convergence}
C.~N. Mavridis and J.~S. Baras, ``Convergence of stochastic vector quantization
  and learning vector quantization with bregman divergences,''
  \emph{IFAC-PapersOnLine}, vol.~53, no.~2, 2020.

\bibitem{mavridis2022online}
------, ``Online deterministic annealing for classification and clustering,''
  \emph{IEEE Transactions on Neural Networks and Learning Systems}, 2022.

\bibitem{biehl2016prototype}
M.~Biehl, B.~Hammer, and T.~Villmann, ``Prototype-based models in machine
  learning,'' \emph{Wiley Interdisciplinary Reviews: Cognitive Science},
  vol.~7, no.~2, pp. 92--111, 2016.

\bibitem{mavridis2022risk}
C.~Mavridis, E.~Noorani, and J.~S. Baras, ``Risk sensitivity and entropy
  regularization in prototype-based learning,'' in \emph{2022 30th
  Mediterranean Conference on Control and Automation (MED)}.\hskip 1em plus
  0.5em minus 0.4em\relax IEEE, 2022, pp. 194--199.

\bibitem{martin_topologyPreservationInSOM_2008}
E.~A. Uriarte and F.~D. Mart{\'\i}n, ``Topology preservation in som,''
  \emph{International journal of applied mathematics and computer sciences},
  vol.~1, no.~1, pp. 19--22, 2005.

\bibitem{saralajew2019robustness}
S.~Saralajew, L.~Holdijk, M.~Rees, and T.~Villmann, ``Robustness of generalized
  learning vector quantization models against adversarial attacks,'' in
  \emph{International Workshop on Self-Organizing Maps}.\hskip 1em plus 0.5em
  minus 0.4em\relax Springer, 2019, pp. 189--199.

\bibitem{rose1998deterministic}
K.~Rose, ``Deterministic annealing for clustering, compression, classification,
  regression, and related optimization problems,'' \emph{Proceedings of the
  IEEE}, vol.~86, no.~11, pp. 2210--2239, 1998.

\bibitem{hinton2006fast}
G.~E. Hinton, S.~Osindero, and Y.-W. Teh, ``A fast learning algorithm for deep
  belief nets,'' \emph{Neural computation}, vol.~18, no.~7, pp. 1527--1554,
  2006.

\bibitem{robbins1951stochastic}
H.~Robbins and S.~Monro, ``A stochastic approximation method,'' \emph{The
  annals of mathematical statistics}, pp. 400--407, 1951.

\bibitem{borkar2009stochastic}
V.~S. Borkar, \emph{Stochastic approximation: a dynamical systems
  viewpoint}.\hskip 1em plus 0.5em minus 0.4em\relax Springer, 2009, vol.~48.

\bibitem{banerjee2005clustering}
A.~Banerjee, S.~Merugu, I.~S. Dhillon, and J.~Ghosh, ``Clustering with bregman
  divergences,'' \emph{Journal of machine learning research}, vol.~6, no. Oct,
  pp. 1705--1749, 2005.

\bibitem{villmann_onlineDLVQmath_2010}
T.~Villmann, S.~Haase, F.-M. Schleif, B.~Hammer, and M.~Biehl, ``The
  mathematics of divergence based online learning in vector quantization,'' in
  \emph{IAPR Workshop on Artificial Neural Networks in Pattern
  Recognition}.\hskip 1em plus 0.5em minus 0.4em\relax Springer, 2010, pp.
  108--119.

\bibitem{rumelhart1986learning}
D.~E. Rumelhart, G.~E. Hinton, and R.~J. Williams, ``Learning representations
  by back-propagating errors,'' \emph{nature}, vol. 323, no. 6088, pp.
  533--536, 1986.

\bibitem{bottou1998online}
L.~Bottou, ``Online learning and stochastic approximations,'' \emph{On-line
  learning in neural networks}, vol.~17, no.~9, p. 142, 1998.

\bibitem{watkins1992q}
C.~J. Watkins and P.~Dayan, ``Q-learning,'' \emph{Machine learning}, vol.~8,
  no. 3-4, pp. 279--292, 1992.

\bibitem{tsitsiklis1994asynchronous}
J.~N. Tsitsiklis, ``Asynchronous stochastic approximation and q-learning,''
  \emph{Machine learning}, vol.~16, no.~3, pp. 185--202, 1994.

\bibitem{borkar2000ode}
V.~S. Borkar and S.~P. Meyn, ``The ode method for convergence of stochastic
  approximation and reinforcement learning,'' \emph{SIAM Journal on Control and
  Optimization}, vol.~38, no.~2, pp. 447--469, 2000.

\bibitem{borkar1997stochastic}
V.~S. Borkar, ``Stochastic approximation with two time scales,'' \emph{Systems
  \& Control Letters}, vol.~29, no.~5, pp. 291--294, 1997.

\bibitem{Kohonen1995}
T.~Kohonen, \emph{Learning Vector Quantization}.\hskip 1em plus 0.5em minus
  0.4em\relax Berlin, Heidelberg: Springer Berlin Heidelberg, 1995, pp.
  175--189.

\bibitem{jaynes1957information}
E.~T. Jaynes, ``Information theory and statistical mechanics,'' \emph{Physical
  review}, vol. 106, no.~4, p. 620, 1957.

\bibitem{hocking2011clusterpath}
T.~D. Hocking, A.~Joulin, F.~Bach, and J.-P. Vert, ``Clusterpath an algorithm
  for clustering using convex fusion penalties,'' in \emph{28th international
  conference on machine learning}, 2011, p.~1.

\bibitem{devroye2013}
L.~Devroye, L.~Gy{\"o}rfi, and G.~Lugosi, \emph{A probabilistic theory of
  pattern recognition}.\hskip 1em plus 0.5em minus 0.4em\relax Springer Science
  \& Business Media, 2013, vol.~31.

\bibitem{sato1996generalized}
A.~Sato and K.~Yamada, ``Generalized learning vector quantization,'' in
  \emph{Advances in neural information processing systems}, 1996, pp. 423--429.

\bibitem{mavridis2021vector}
C.~N. Mavridis and J.~S. Baras, ``Vector quantization for adaptive state
  aggregation in reinforcement learning,'' in \emph{2021 American Control
  Conference (ACC)}.\hskip 1em plus 0.5em minus 0.4em\relax IEEE, 2021, pp.
  2187--2192.

\bibitem{mavridis2021maximum}
C.~N. Mavridis, N.~Suriyarachchi, and J.~S. Baras, ``Maximum-entropy
  progressive state aggregation for reinforcement learning,'' in \emph{2021
  60th IEEE Conference on Decision and Control (CDC)}.\hskip 1em plus 0.5em
  minus 0.4em\relax IEEE, 2021, pp. 5144--5149.

\bibitem{mavridis2022sparse}
C.~N. Mavridis, G.~Kontoudis, and J.~S. Baras, ``Sparse gaussian process
  regression using progressively growing learning representations,'' in
  \emph{2022 61st IEEE Conference on Decision and Control (CDC)}.\hskip 1em
  plus 0.5em minus 0.4em\relax IEEE, 2022.

\bibitem{haarnoja2017}
T.~Haarnoja, H.~Tang, P.~Abbeel, and S.~Levine, ``{Reinforcement Learning with
  Deep Energy-Based Policies},'' in \emph{Proceedings of the 34th International
  Conference on Machine Learning - Volume 70}, ser. ICML'17.\hskip 1em plus
  0.5em minus 0.4em\relax JMLR.org, 2017, p. 1352–1361.

\bibitem{schulman2017proximal}
J.~Schulman, F.~Wolski, P.~Dhariwal, A.~Radford, and O.~Klimov, ``{Proximal
  Policy Optimization Algorithms},'' \emph{arXiv preprint arXiv:1707.06347},
  2017.

\bibitem{noorani2023risk}
E.~Noorani, C.~Mavridis, and J.~Baras, ``Risk-sensitive reinforcement learning
  with exponential criteria,'' \emph{arXiv}, 2022.

\bibitem{aLaVigna_LVQconvergence_1990}
J.~S. Baras and A.~LaVigna, ``Convergence of a neural network classifier,'' in
  \emph{Advances in Neural Information Processing Systems}, 1991, pp. 839--845.

\bibitem{smith1988using}
J.~W. Smith, J.~Everhart, W.~Dickson, W.~Knowler, and R.~Johannes, ``Using the
  adap learning algorithm to forecast the onset of diabetes mellitus,'' in
  \emph{Proceedings of the annual symposium on computer application in medical
  care}.\hskip 1em plus 0.5em minus 0.4em\relax American Medical Informatics
  Association, 1988, p. 261.

\bibitem{bottou1995convergence}
L.~Bottou and Y.~Bengio, ``Convergence properties of the k-means algorithms,''
  in \emph{Advances in neural information processing systems}, 1995, pp.
  585--592.

\bibitem{mavridis2021progressive}
C.~N. Mavridis and J.~S. Baras, ``Progressive graph partitioning based on
  information diffusion,'' in \emph{2021 60th IEEE Conference on Decision and
  Control (CDC)}.\hskip 1em plus 0.5em minus 0.4em\relax IEEE, 2021, pp.
  37--42.

\bibitem{barto1983neuronlike}
A.~G. Barto, R.~S. Sutton, and C.~W. Anderson, ``Neuronlike adaptive elements
  that can solve difficult learning control problems,'' \emph{IEEE transactions
  on systems, man, and cybernetics}, no.~5, pp. 834--846, 1983.

\end{thebibliography}


\vspace{-3em}

\begin{IEEEbiography}[{\includegraphics[width=1in,height=1.25in,clip,keepaspectratio]
{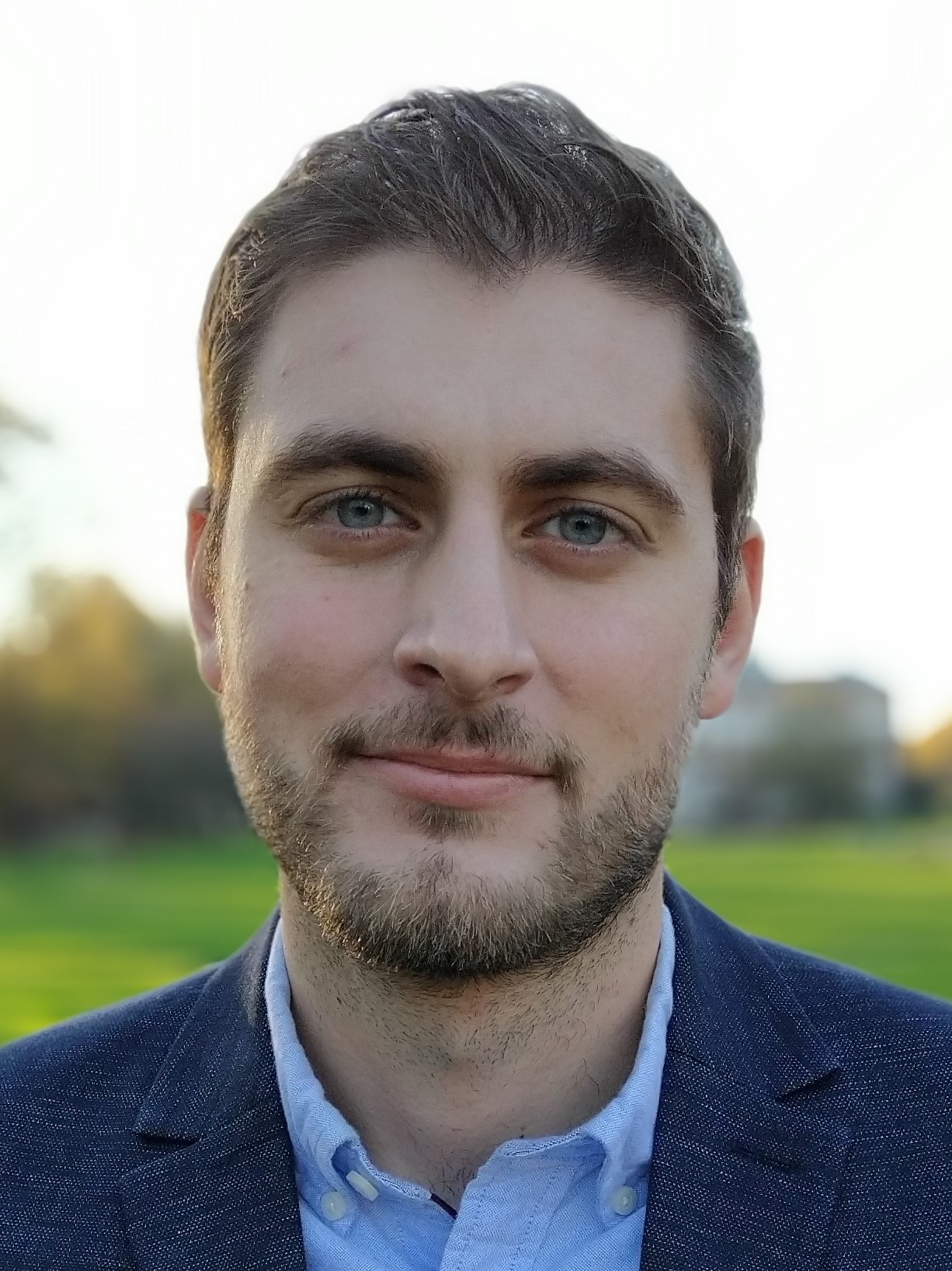}}]{Christos N. Mavridis} (M'20) 
received the Diploma degree in electrical and computer engineering from the National Technical University of Athens, Greece, in 2017,
and the M.S. and  Ph.D. degrees in electrical and computer engineering at the University of Maryland, College Park, MD, USA, in 2021. 
His research interests include systems and control theory, stochastic optimization,  
learning theory, multi-agent systems, and robotics. 

He has served as a postdoctoral associate at the University of Maryland, and a visiting postdoctoral fellow at KTH Royal Institute of Technology, Stockholm. He has worked as a research intern for the Math and Algorithms Research Group at Nokia Bell Labs, NJ, USA, and the System Sciences Lab at Xerox Palo Alto Research Center (PARC), CA, USA. 

Dr. Mavridis is an IEEE member, and a member of the Institute for Systems Research (ISR) and the Autonomy, Robotics and Cognition (ARC) Lab. He received the Ann G. Wylie Dissertation Fellowship in 2021, and the A. James Clark School of Engineering Distinguished Graduate Fellowship, Outstanding Graduate Research Assistant Award, and Future Faculty Fellowship, in 2017, 2020, and 2021, respectively. He has been a finalist in the Qualcomm Innovation Fellowship US, San Diego, CA, 2018, and he has received the Best Student Paper Award (1st place) in the IEEE International Conference on Intelligent Transportation Systems (ITSC), 2021.
\end{IEEEbiography}

\vspace{-4em}

\begin{IEEEbiography}[{\includegraphics[width=1in,height=1.25in,clip,keepaspectratio]
{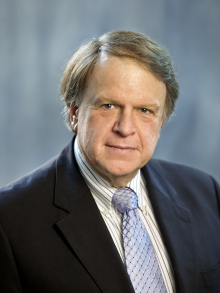}}]{John S. Baras} (LF'13) 
received the Diploma degree in electrical and mechanical engineering from the National Technical University of Athens, Athens, Greece, in 1970, and the M.S. and Ph.D. degrees in applied mathematics from Harvard University, Cambridge, MA, USA, in 1971 and 1973, respectively.

He is a Distinguished University Professor and holds the Lockheed Martin Chair in Systems Engineering, with the Department of Electrical and Computer Engineering and the Institute for Systems Research (ISR), at the University of Maryland College Park. From 1985 to 1991, he was the Founding Director of the ISR. Since 1992, he has been the Director of the Maryland Center for Hybrid Networks (HYNET), which he co-founded. His research interests include systems and control, optimization, communication networks, applied mathematics, machine learning, artificial intelligence, signal processing, robotics, computing systems, security, trust, systems biology, healthcare systems, model-based systems engineering.

Dr. Baras is a Fellow of IEEE (Life), SIAM, AAAS, NAI, IFAC, AMS, AIAA, Member of the National Academy of Inventors and a Foreign Member of the Royal Swedish Academy of Engineering Sciences. Major honors include the 1980 George Axelby Award from the IEEE Control Systems Society, the 2006 Leonard Abraham Prize from the IEEE Communications Society, the 2017 IEEE Simon Ramo Medal, the 2017 AACC Richard E. Bellman Control Heritage Award, the 2018 AIAA Aerospace Communications Award. In 2016 he was inducted in the A. J. Clark School of Engineering Innovation Hall of Fame. In 2018 he was awarded a Doctorate Honoris Causa by his alma mater the National Technical University of Athens, Greece.   
\end{IEEEbiography}

\end{document}